\documentclass[12pt]{article}

\usepackage{hyperref}
\usepackage{graphicx}% Include figure files
\usepackage{amsmath}
\usepackage{amsthm}
\usepackage{amssymb}
\usepackage{tipa}
\usepackage{subfig} % caution: not compatible with revtex4
\usepackage{xcolor}
\usepackage{mathrsfs}
\usepackage{physics}
\theoremstyle{plain}

\usepackage{tikz}
\usepackage{pgfplots}
\usepackage{tikz-cd}
\usetikzlibrary{arrows}
\usetikzlibrary{intersections}
\usetikzlibrary{shapes.geometric}
\usetikzlibrary{decorations.pathmorphing, patterns,shapes}

\usepackage[nosort]{cite}

\newlength{\abstractwidth}
\renewcommand{\title}[1]{\vbox{\center\bf{\Large{#1}}}\vspace{5mm}}
\renewcommand{\author}[1]{\vbox{\center#1}\vspace{5mm}}
\newcommand{\address}[1]{\vbox{\center\em#1}}

\newtheorem{lemma}{Lemma}[section]

\begin{document}

\begin{titlepage}
\begin{center}
\hfill \\
\hfill \\
\vskip 1cm

\title{Butterfly velocity and  bulk causal structure} 
\author{Xiao-Liang Qi$^1$, Zhao Yang$^1$}
\address{
$^1$Department of Physics, Stanford University, Stanford, CA 94305, USA}

\date{\today}
\end{center}
  
\begin{abstract}
The butterfly velocity was recently proposed as a characteristic velocity of chaos propagation in a local system. Compared to the Lieb-Robinson velocity that bounds the propagation speed of all perturbations, the butterfly velocity, studied in thermal ensembles, is an "effective" Lieb-Robinson velocity for a subspace of the Hilbert space defined by the microcanonical ensemble. In this paper, we generalize the concept of butterfly velocity beyond the thermal case to a large class of other subspaces. Based on holographic duality, we consider the code subspace of low energy excitations on a classical background geometry. Using local reconstruction of bulk operators, we prove a general relation between the boundary butterfly velocities (of different operators) and the bulk causal structure. Our result has implications in both directions of the bulk-boundary correspondence. Starting from a boundary theory with a given Lieb-Robinson velocity, our result determines an upper bound of the bulk light cone starting from a given point. Starting from a bulk space-time geometry, the butterfly velocity can be explicitly calculated for all operators that are the local reconstructions of bulk local operators. If the bulk geometry satisfies Einstein equation and the null energy condition, for rotation symmetric geometries we prove that infrared operators always have a slower butterfly velocity that the ultraviolet one. For asymptotic AdS geometries, this also implies that the butterfly velocities of all operators are upper bounded by the speed of light. We further prove that the butterfly velocity is equal to the speed of light if the causal wedge of the boundary region coincides with its entanglement wedge. Finally, we discuss the implication of our result to geometries that are not asymptotically AdS, and in particular, obtain constraints that must be satisfied by a dual theory of flat space gravity.

%From
%  the boundary to bulk point of view, we show that the Lieb-Robinson
%  velocity on the boundary (boundary locality) determines an upper
%  bound of the speed of light in the bulk (bulk locality) in a
%  holographic mapping with error correction properties, such as the
%  random tensor network models. From the bulk to boundary point of
%  view, we derive the properties of boundary butterfly velocities from
%  a known bulk space-time geometry. Specifically, by assuming the
%  geometry to satisfy the Einstein equation and the null energy
%  condition, we prove one monotonicity of the butterfly velocities of
%  boundary operators from UV to IR. 

\end{abstract}
\end{titlepage}

\tableofcontents

\baselineskip=17.63pt

\section{Introduction}
% [\ZY{I move this paragraph to the beginning}.]
The butterfly velocity\cite{shenker2013black, roberts2014localized} is
generally a {\it state dependent} measure of the quantum many-body
dynamics that quantifies the propagation velocity of the causal
influence for a local perturbation {\it when the influence is probed
  in a subspace of many-body states}. For a relativistic system, the
velocity of causal influence is the speed of light if we study a
generic perturbation acting on arbitrary quantum states. However, the
butterfly velocity at finite temperature is generically slower than
the speed of light because we are probing the causal influence only in
states with a fixed energy density. The butterfly velocity can be
measured by the effective size of commutator between local
operators. In large $N$ theories with a semi-classical holographic
dual, the commutator is believed to behaves as
$ \langle [W(x,t),V(0,0)]^2\rangle_\beta = \frac{C}{N^2}
e^{\lambda_L(t-x/v_B)} +O(N^{-4})$, where $W,V$ are the generic
operators, $C$ is a constant, $\lambda_L$ is the Lyapunov exponent,
and $v_B$ is the butterfly velocity.  $\langle\rangle_\beta$
represents the thermal average at temperature $1/\beta$. Moreover,
according to the anti-de Sitter/conformal field theory (AdS/CFT)
correspondence \cite{maldacena1999large, witten1998anti,
  gubser1998gauge}, the large $N$ gauge theory at finite temperature
is dual to the AdS black hole geometry, and applying boundary
operators onto the thermal ensemble corresponds to shooting shock
waves into the black hole in the bulk \cite{shenker2013black,
  roberts2014localized}. Thus the butterfly velocity can be calculated
holographically in the bulk by evaluating the back-reaction of the
shock wave geometry\cite{shenker2013black, roberts2014localized,
  roberts2016lieb}. Another independent calculation
\cite{mezei2016entanglement} of the butterfly velocity was
accomplished by studying the expansion rate of the extremal surface
near the black hole horizon. This result is consistent with previous
shockwave calculations.  However, all the discussions on butterfly
  velocities so far are about the thermal ensembles. Also, as will be
  seen from the results of the present paper, the butterfly velocities
  calculated in these previous works only correspond to those bulk
  operators close to the horizon.

Motivated by these works, our paper aims to generalize the concept of
butterfly velocity beyond the thermal ensemble to more generic
subspaces of states. We observe that in the same code-subspace
different operators generically have different velocities, and we
establish a concrete relation between the boundary butterfly velocity
and the bulk causal structure in holographic systems.  The connection
between these two ends is the quantum error correction conditions,
which have been observed both within the AdS/CFT
correspondence\cite{almheiri2014bulk,dong2016reconstruction,harlow2016ryu,dong2016bulk}
and in the tensor network
models\cite{pastawski2015holographic,yang2016bidirectional,hayden2016holographic,qi2017holographic,donnelly2016living}. % \XLQ{[consider
  % to add Don Marolf's paper on generalizing pentagon code.]}
A natural generalization of the thermal ensemble is the code subspace
, {\it i.e.}, the subspace of small fluctuations around a classical
geometry. The code subspace for an AdS black hole includes states with
a fixed energy density from a (microcanonical) thermal ensemble of the
boundary. %. that contains all the thermal states within a narrow energy window.
Holographically, these states are dual to the same classical
black-hole geometry.  However, they are different black-hole
micro-states or have different states in the effective field theory
living on top of the black hole geometry. We demonstrate that the
concept of the butterfly velocity can be generalized to arbitrary code
subspaces, and it is defined as the propagation velocity of
  certain operators measured by the states in the code subspace. For
large $N$ theories with a semi-classical bulk dual, using the
entanglement wedge reconstruction, we prove that the bulk causal
structure determines the butterfly velocities of operators on the
boundary and vise versa. Our generalization of the butterfly velocity
not only reproduces all the results in the thermal ensembles, but also
allows us to predict the boundary butterfly velocities in more general
boundary states with classical bulk dual geometries, such as
time-dependent geometries. {Our generalized butterfly velocity is
  operator-specific. For local operators in the bulk that are mapped
  to boundary operators in the same disk-shape region, we prove that
  butterfly velocity of the operator deeper in the bulk is slower for
  any geometry that satisfied Einstein's equation (EE) and the null
  energy condition (NEC). In addition, using tensor network
  construction, one can even construct spatial geometries that are not
  asymptotically AdS while still preserving error correction
  properties of the bulk-boundary correspondence.  Our discussion also
  applies for such geometries and impose constraints on the possible
  dual theories.} {In particular, we study the example of a flat
  geometry, and show that the boundary theory has to have a divergent
  butterfly velocity, which therefore has to be a nonlocal theory.}

%they are not supposed to be asymptotic AdS; and they can even be constructed using tensor networks. The only requirement is that the quantum error correction conditions are satisfied, such that the entanglement wedge reconstruction is available.  However, the more assumptions we make on the geometries, the more constrains we obtain for the butterfly velocities.

The remainder of the paper is organized as follows.  In
Sec.\ref{sec:def}, we give the precise definition of the butterfly
velocity, for a given operator and a given code subspace. % \XLQ{[I
% think mentioning large N is a bit confusing here. It should be ok to
% keep the statement general.]}
Then in Sec.\ref{sec:TN}, we analyze the indication of this definition
to the bulk theory. We show that in a holographic mapping with error
correction properties, such as the random tensor network
models\cite{hayden2016holographic,qi2017holographic}, the
Lieb-Robinson velocity of the boundary, which is the upper bound of
the butterfly velocities, determines an upper bound of the speed of
light in the bulk. Running this argument on different regions of the
boundary, we obtain a bulk region that encloses the casual future of a
bulk point. The physical interpretation of this result is that the
quantum error correction properties and local boundary dynamics imply
the local bulk dynamics in the code subspace. Specifically, in the
random tensor network models, this result means that the holographic
mapping defined by random tensor networks always maps local boundary
dynamics to local bulk dynamics in the code subspace.
  % not only establish a static correspondence between the bulk states
  % and the boundary states, but also map local boundary theories into
  % local bulk theories in the code subspace.
In sections \ref{sec:protocol}-\ref{sec:flat}, we focus on the
bulk-to-boundary direction and derive the properties of boundary
butterfly velocities from a known bulk space-time geometry. In
Sec.\ref{sec:protocol}, we show that in the holographic theory the
light cone at the bulk point $x$ determines the butterfly velocity of
the boundary operators that are local reconstruction of a local bulk
operator at $x$. In Sec.\ref{sec:monobv}, we assume the bulk geometry
to satisfy the Einstein equation (EE) and the null energy condition
(NEC) and conclude that, roughly speaking, the butterfly velocities of
the boundary operators decrease monotonically from the UV to IR. The
precise statement can be found in Sec.\ref{sec:monobv}. This
conclusion also implies that if the bulk geometry is asymptotically
AdS, the butterfly velocities of boundary operators are upper bounded
by the speed of the light.  Furthermore, in Sec.\ref{sec:causalwedge},
we prove that, if the bulk geometry is asymptotic AdS that satisfies
EE and NEC, and if the causal wedge of the boundary region $A$
coincides with its entanglement wedge, then the butterfly velocities
of the boundary operators that are supported on the whole boundary
region $A$ saturate the speed of the light. In Sec.\ref{example}, we
show some explicit calculations of the butterfly velocity for simple
bulk geometries, such as $d+1$-dimensional pure AdS space, the 3d
Banados, Teitelboim and Zanelli (BTZ) black hole
\cite{banados1992black}, % \XLQ{[give name of BTZ]}
and $d+1$ ($d>2$) AdS Schwarzschild black
holes. %Our method does not make explicit assumptions on the bulk geometries as long as the entanglement wedge reconstruction is satisfied. Thus it allows us to talk about the physics beyond the conventional AdS/CFT correspondence.
In Sec.\ref{sec:flat}, we go beyond AdS/CFT correspondence and study
the consequence of our results for a flat bulk geometry. Our results
lead to necessary conditions for a boundary theory to have a flat bulk
dual. Finally we conclude in Sec.\ref{sec:conclusion}.

\section{Definition of the butterfly velocity}\label{sec:def}
Previous discussions on the butterfly velocities are based on the
assumption that %in the large $N$ theory commutators evaluated in
the expectation value of commutator square of two operators in a thermal ensemble has the following behavior
$ \langle [W(x,t),V(0,0)]^2\rangle_\beta \propto
e^{\lambda_L(t-x/v_B)}$, in which case $v_B$ is called the butterfly
velocity. If we use a microcanonical ensemble instead of the canonical one, the quantity $\langle [W(x,t),V(0,0)]^2\rangle_\beta$ is the two norm of the commutator in the subspace of states in the microcanonical ensemble. %However, in order to discuss the butterfly velocity in the general settings, we must generalize the definition of butterfly velocities.
In the following, we will generalize the definition of butterfly velocities from thermal ensemble to more general subspaces. The generalized butterfly velocity describes the maximal information propagation velocity in a given subspace of the Hilbert space. 

We set the notations as follows. For a system with locality, such as a system defined in a Riemann manifold, or a discrete system defined on a graph, we define two regions $A$ and $B$, and denote $d(B,A)$ to be the distance between them, defined as the minimum of distance between two points $x\in A,~y\in B$. We denote $\{R|d(R,A)=D\}$ as all the
regions whose distance to $A$ is $D$. Besides, if an operator $O_A$
commutes with all operators (in Heisenberg picture) in the region $B$ at time $t$
sandwiched by two states $|\psi_i\rangle$, $|\psi_j\rangle$, then we abbreviate
this relation as
\begin{equation}
  \label{eq:3}
  \langle\psi_i| [O_A, B(t)]|\psi_j\rangle = 0
\end{equation}

Now we consider a system with holographic duality and start from the
boundary system. Given a boundary code subspace $\mathcal{H}_c$ (which
at this moment can be any subspace of the Hilbert space
$\mathcal{H}$), and a generic boundary operator $O_A$ with support on
region $A$, we define the butterfly velocity $v(O_A;\mathcal{H}_c)$ as
the minimal velocity such that in the $\Delta t \rightarrow 0$ limit,
$\forall ~ B\in \{R|d(R,A)= v(O_A;\mathcal{H}_c) \Delta t\}$,
$\forall ~|\psi_i\rangle, |\psi_j\rangle \in \mathcal{H}_c$,
\begin{equation}\label{def}
  \langle\psi_i| [O_A, B(\Delta t)]|\psi_j\rangle = 0
\end{equation}

The precise meaning of $\Delta t\rightarrow 0$ limit in the
$\epsilon, \delta$ language is elaborated in Appendix.\ref{SecDef}. In fact, it is more realistic to require the commutator to be small, controlled by a small parameter in the system, rather than exactly vanish. For large $N$ theories with a semiclassical dual, the small parameter is $\frac1{N}$, and the transition from zero and nonzero commutator at the ``butterfly cone" is sharp in the large $N$ limit. All our discussion below applies to such large $N$ limit.
%[In the
%  example of holographic models, the commutator is suppressed by $1/N$
%  for any separation. I guess you mean an extra power of $1/N$? Is
%  there a sharp transition at the butterfly cone? (If yes, this is a
%  significant difference from the SYK case.) ] } \ZY{[I am not sure
%  about the SYK, but the large $N$ behavior is as what we wrote in the
%  introduction
%  $ \langle [W(x,t),V(0,0)]^2\rangle_\beta = \frac{C}{N^2}
%  e^{\lambda_L(t-x/v_B)} +O(N^{-4})$. Thus there is a sharp transition
%  if $N$ is big.]}

Intuitively, the definition above means that at a small time $\Delta t$, operator $O_A$ evolve to a Heisenberg operator $O_A(\Delta t)$ that is ``effectively" supported in a region that is a slight expansion of $A$ by the distance $v(O_A;\mathcal{H}_c)\Delta t$. In other words, $O_A$ still commutes with all operators in the complement of this small expansion of $A$ after the Heisenberg evolution, {\it if the commutator operator is only measured in the code subspace $\mathcal{H}_c$}.  First of all, we
know that as long as $v(O_A;\mathcal{H}_c)$ is big enough, it is
obvious that Eq.\ref{def} is satisfied. For example, if
$v(O_A;\mathcal{H}_c) > c$, then obviously all the operators in the
boundary region $B\in \{R|d(R,A)=v(O_A;\mathcal{H}_c)\Delta t\}$ at
time $\Delta t$ commutes with $O_A$. Thus we minimize
$v(O_A;\mathcal{H}_c)$ to find the butterfly velocity.

{The definition above applies to any subspace $\mathcal{H}_c$ of the boundary Hilbert space, although what we will be interested in are the subspaces in which the bulk-to-boundary isometry with error correction properties is defined.\cite{almheiri2014bulk} It should be clarified that both operators that appear in the commutator ($O_A$ and generic operators $O_B(\Delta t)$ in $B(\Delta t)$) act on the whole Hilbert space, although the butterfly velocity only measures the norm of the commutator in the code subspace $\mathcal{H}_c$. Actually, for the code subspace we are interested in, all operators that only act in the code subspace are global on the boundary, so that any local operator like $O_A$ is necessarily coupling the code subspace with its complement in the full Hilbert space. }

%defined for some code subspace $\mathcal{H}_c$, $O_A$ is not necessary to be an operator that only lives in the code subspace, which means in general $O_A \neq P_cO_AP_c$. Also, we need to involve operators beyond the code subspace in $B$ to decide the butterfly velocity, which can be seen from the definition.

\section{Boundary time evolution and bulk causality}\label{sec:TN}

The relation between the butterfly velocity on the boundary and the
bulk causal structure is bidirectional. In this section, we look at
the direction from the boundary to the bulk. We will show that the boundary butterfly velocities and local reconstruction properties together give bounds on the causal structure of the bulk dual theory. 

\subsection{An overview of the local reconstruction property}

Our results apply generally to systems with a holographic operator correspondence between bulk and boundary with local reconstruction properties, which include the standard AdS/CFT systems \cite{almheiri2014bulk,dong2016reconstruction,harlow2016ryu,dong2016bulk} and holographic mappings defined by tensor networks \cite{pastawski2015holographic,yang2016bidirectional,hayden2016holographic}. For concreteness, we briefly overview the local reconstruction property and code subspace in the random tensor networks studied in Ref. \cite{hayden2016holographic}. The readers who are already familiar with local reconstruction can skip this subsection. 

A random tensor network with dangling legs in both bulk and boundary, as is illustrated in Fig. \ref{fig:RTN}, defines a linear map between the bulk and boundary Hilbert spaces. Each tensor can be considered as an operator $V_x$ with matrix element ${V_x}^a_{;\alpha\beta\gamma}$, which maps the bulk state $a$ to in-plane states $\alpha\beta\gamma$. Then the contraction of internal lines is equivalent to projecting the state of the two ends of a link into a maximally entangled state $|xy\rangle$. The holographic mapping from bulk to boundary is defined as an operator
\begin{eqnarray}
M=\prod_{\langle xy\rangle}\langle xy|\prod_x V_x
\end{eqnarray}
We denote the dimension of bulk vertex index $a$ as $D_b$ and that of the boundary and the internal indices as $D$. If there are $V$ vertices in the bulk and $V_B$ vertices on the boundary, the bulk Hilbert space dimension is $D_b^V$ and that of the boundary is $D^{V_B}$. Ref. \cite{hayden2016holographic} proves that the mapping $M$ is an isometry in the limit $D\rightarrow \infty$ with $D_b$ finite. In this case the bulk Hilbert space is mapped by $M$ to a subspace of the boundary Hilbert space, which is the code subspace $\mathcal{H}_c$ we are interested in here. In addition to the isometry property, Ref. \cite{hayden2016holographic} also proves the following local reconstruction property, similar to that in AdS/CFT. Each boundary region $A$ is associated with a minimal area surface $\gamma_A$ bounding it, and the region between $A$ and $\gamma_A$ is called the entanglement wedge of $A$, denoted by $E_A$.\footnote{More precisely, $E_A$ here corresponds to a spatial slice in the space-time entanglement wedge in the AdS/CFT case.} For an operator in the bulk $\phi$ with support in a bulk region $R$, if we choose a boundary region $A$ such that $R\subset E_A$, then there exists a boundary operator $O_A$ supported in region $A$, such that $O_A|\psi\rangle=M\phi M^\dagger |\psi\rangle$ for any state $|\psi\rangle\in\mathcal{H}_c$. Since a bulk region $R$ can be enclosed by entanglement wedges of different boundary regions, there are multiple boundary operators which reconstruct the same bulk operator in the code subspace. These different boundary operators are clearly different operators in the entire boundary Hilbert space, but their matrix elements are identical when acting on the code subspace states. 
\begin{figure}[!htb]
  \centering
  \includegraphics[width=0.5\textwidth]{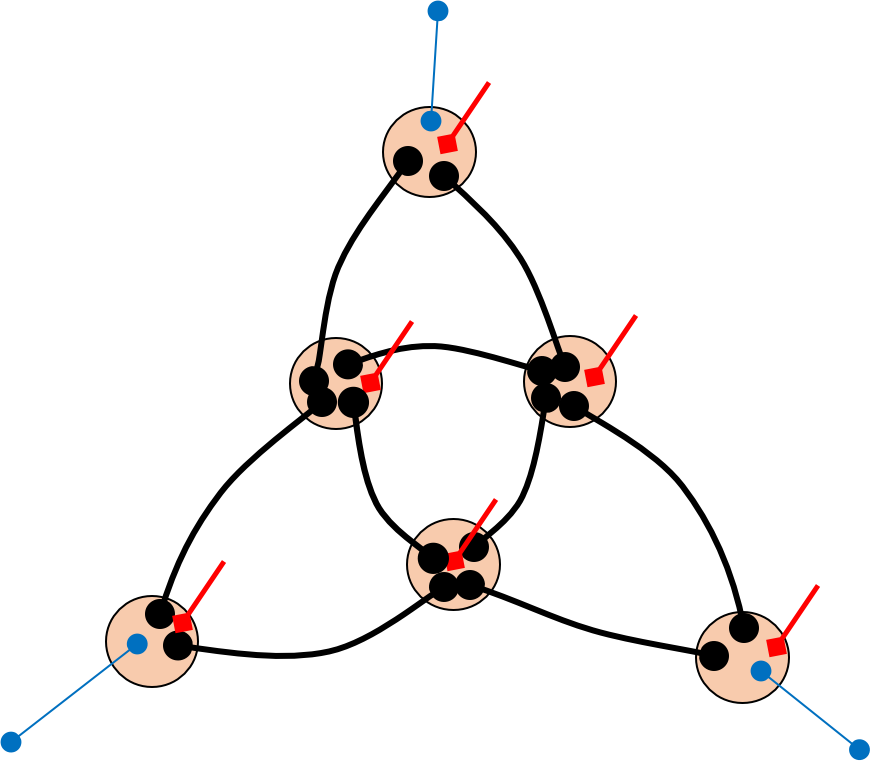}
  \caption{Illustration of a random tensor network, which defines a linear map between the bulk (red legs) and the boundary (blue legs). }
 \label{fig:RTN}
\end{figure}

The definition of code subspace and local reconstruction properties in AdS/CFT is the same as the tensor network case reviewed above, with the graph geometry replaced by a Riemann geometry. The bulk operators in the code subspace are quantum field theory operators with energy of order $1$ which do not change the background geometry to the leading order of Newton constant $G_N$. 

\subsection{Bound on bulk light cone}

In general, it is complicated to decide the butterfly velocity for all
operators. However, the butterfly velocity is upper bounded by
the speed of light $c$ for a Lorentz invariant theory, or the
Lieb-Robinson velocity $v_{LR}$ for a lattice system
\cite{lieb1972finite,nachtergaele2006propagation,hastings2010locality}. Using
this upper bound which we denote as $v_{LR}$, we can determine the
upper bound of the light-cone size in the
bulk. %An explicit application is the holographic code defined using the random tensor network \cite{hayden2016holographic}.

%The holographic code defined using the random tensor network is
%constructed as following. The contracted links and the un-contracted
%links on the boundary in the random tensor network are EPR pairs whose
%bond dimension are $D$. The ends of the links on each graph vertex are
%projected onto a random pure state. Besides on every vertex, there is
%an un-contracted link whose bond dimension is $D_b$. In the limit
%$D_b<< D$, the holographic code defines an isometry from the
%un-contracted bulk links to the un-contracted boundary links.  

For a bulk point $x$, as is shown in Fig. \ref{fig:envelope} (a), we consider a boundary region $A$ such that the entanglement wedge $E_A$ barely includes $x$. More precisely, $x\in E_A$ and $x$ is a distance $\epsilon$ away from the minimal surface $\gamma_A$, with $\epsilon\rightarrow 0^+$. According to local reconstruction, for any bulk local operator $\phi_x$ acting at site $x$, there is a corresponding boundary operator $O_A$ supported on $A$ which is the reconstruction of $\phi_x$. If the boundary butterfly velocity for arbitrary operator is upper bounded by $v_{LR}$, after time $\Delta t$ the operator $O_A$ will evolve to some operator that is supported in a slightly bigger region $A_{v_{LR}\Delta t}$, which is an expansion of $A$ defined by the following:
\begin{equation}
  A_{v_{LR}\Delta t} \equiv \{x\in \text{boundary}\left| \exists y\in A,  ~\text{s.t.}~  d(x,y) \leq  v_{LR}\Delta t \right \}
\label{eq:expansionofA}
\end{equation}
In other words, $A_{v_{LR}\Delta t}$ is the union of all balls of size $v_{LR}\Delta t$ (in the boundary) with center position in $A$. Since $O_A(\Delta t)$ is supported in $A_{v_{LR}\Delta t}$, it commutes with all boundary operators in its complement, which we denote as $B=\overline{A_{v_{LR}\Delta t}}$. As a consequence, it also commutes with all bulk operators in the code subspace that are supported in the entanglement wedge of the complement region $E_B$. In other words, the bulk operator $\phi_x$ is evolved by the boundary time evolution to an operator $\phi_x(\Delta t)$ which is supported in the entanglement wedge of $A_{v_{LR}\Delta t}$, denoted as $E_{A_{v_{LR}\Delta t}}=\overline{E_B}$. In short we can denote 
\begin{eqnarray}
\phi_x(\Delta t)\in E_{A_{v_{LR} \Delta t}}
\end{eqnarray}

For a fixed point $x$, one can define an infinite family of minimal surfaces passing $x$, which corresponds to an infinite family of boundary regions. The argument above applies to each such region, so that $\phi_x(\Delta t)$ actually is supported in the intersection of the entanglement wedges $E_{A_{v_{LR}\Delta t}}$ for all these region choices (Fig. \ref{fig:envelope} (b)): 
\begin{eqnarray}
\phi_x(\Delta t)\in D(x,\Delta t)\equiv \bigcap_{A: x\in\gamma_A}E_{A_{v_{LR}\Delta t}}
\end{eqnarray}
This result shows that arbitrary local perturbation at $x$, as long as it is in the code subspace, can only spread in the intersection region $D(x,\Delta t)$ after time $\Delta t$. If the location of minimal surface $\gamma_A$ is a continuous function of the boundary region $A$, the intersection $D(x,\Delta t)$ is a disk which shrinks to point $x$ in the $\Delta t\rightarrow 0$ limit. In a spacetime picture, the union of $D(x,\Delta t)$ for all $\Delta t$ is a spacetime region that the bulk light cone starting from point $x$ must reside in. Taking $\Delta t\rightarrow 0$, we obtain upper bound of bulk speed of light if $D(x,\Delta t)$ shrinks to zero.
\begin{figure}[htb!]
  \centering
  \includegraphics[width=0.8\textwidth]{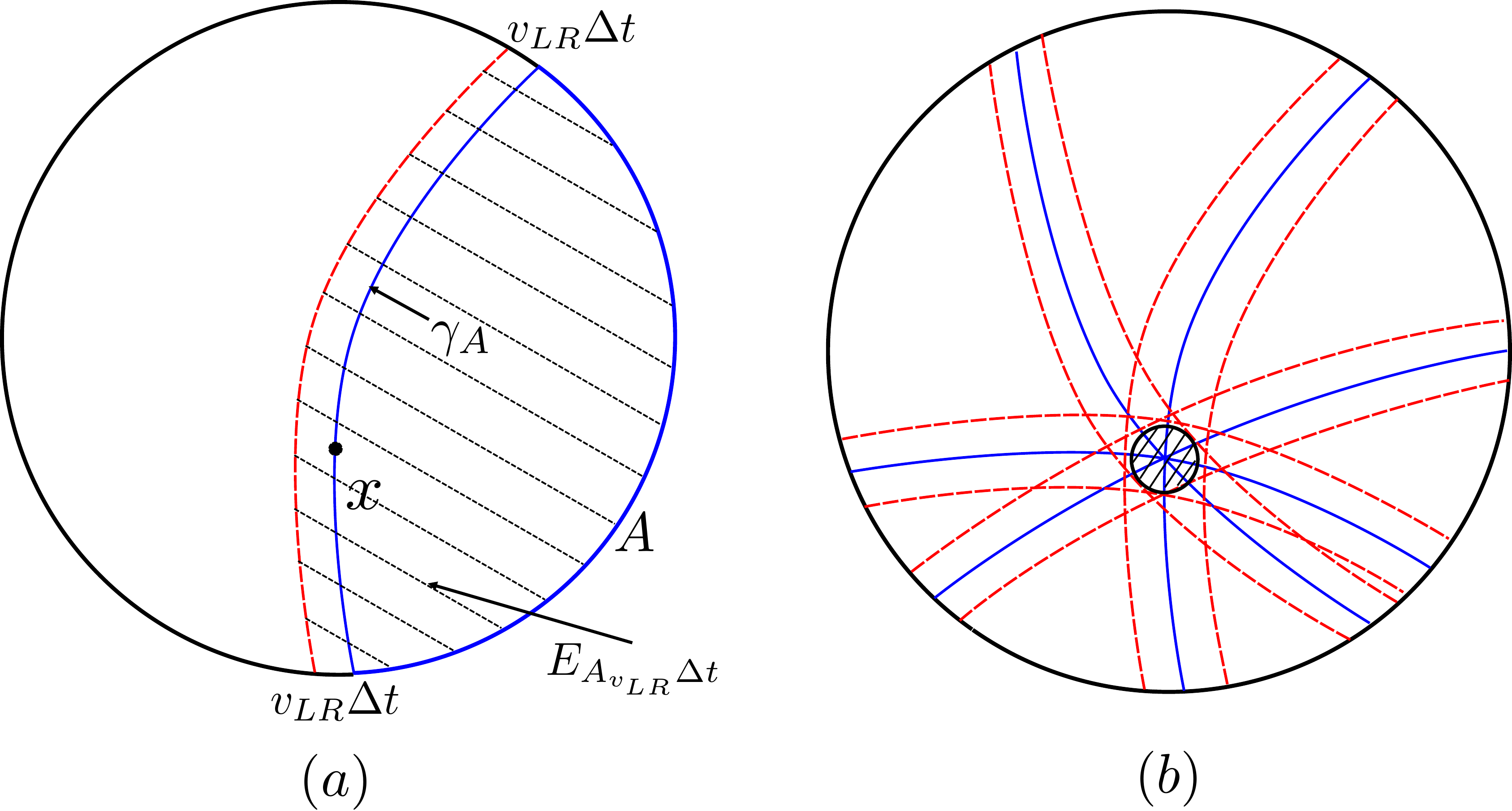}
  \caption{Illustration of the bound on causal future of a point $x$ in the bulk. (a) A bulk local operator $\phi_x$ at a point $x$ on the minimal surface $\gamma_A$ can be reconstructed to a boundary region $A$. The slightly bigger region $A_{v_{LR}\Delta t}$ is defined as the expansion of $A$ by size $v_{LR}\Delta t$ (see text). The shaded region is the entanglement wedge of the bigger region $E_{A_{v_{LR}\Delta t}}$. $\phi_x(\Delta t)$ commutes with all bulk (Schoerdinger) operators outside the shaded region. (b) By repeating the construction in (a) for different boundary regions that reconstruct $\phi_x$, we obtain a domain around $x$ by the intersection of entanglement wedges $E_{A_{v_{LR}\Delta t}}$ for different $A$, which is the upper bound of causal future of $x$ since $\phi_x(\Delta t)$ commutes with all bulk operators outside this region. % \XLQ{[Zhao, please add $E_{A_{v_{LR}\Delta t}}$ and $\gamma_A$ to fig. (a). Also, in future please write the caption following the style of my modification. Your captions look like part of main text. (see also fig 3 caption)]}
  }
  \label{fig:envelope}
\end{figure}

In summary, we have shown that the local reconstruction property and boundary locality (finite $v_{LR}$) put constraints on locality of bulk dynamics. We would like to make a few further comments here. 

Firstly, our starting point in this section is a holographic mapping with error
  correction properties, defined by a spatial geometry. We are always
  assuming the minimal surfaces passing $x$ resides in a particular
  spatial slice, so that RT formula applies. Our result shows that if
  the spatial geometry is time translational invariant, the local
  reconstruction constrains the causal structure in the bulk. Our
  discussion can be generalized to a time-dependent spatial geometry
  (which means we are using different holographic code at different
  time), but that requires the assumption that the extremal surfaces
  bounding boundary regions at a given time all reside in a single
  bulk slice, which is generically not true. 

  Secondly, for geometries with minimal surface experiencing topology
  change, our construction may not apply. For example, in the black
  hole geometries generically there is a region near the horizon
  called ``entanglement shadow" \cite{balasubramanian2014entwinement,
    engelhardt2014extremal}, where no minimal surface bounding any
  boundary region can reach. An operator $\phi_x$ in the entanglement
  shadow can only be reconstructed to a boundary region whose
  entanglement wedge includes the entire black hole and the
  entanglement shadow region. In this case, our construction results
  in a $D(x,\Delta t)$ that remains finite at small $\Delta
  t$. $D(x,\Delta t\rightarrow 0)$ will be the entire entanglement
  shadow. Therefore the locality of physics in the entanglement shadow
  region cannot be understood from local reconstruction. It is an
  interesting open question what is the reason, from boundary
  dynamics, of bulk locality in the entanglement shadow.

In a geometry with entanglement shadow, for points outside the entanglement shadow our construction applies. However, the presence of the entanglement shadow still has a nontrivial effect on the upper bound we obtains. Consider the set of all geodesic surfaces which are minimal surfaces of boundary regions, and include point $x$. As is shown in Fig. \ref{fig:outshadow} (a), in a geometry with entanglement shadow, the normal direction of such surfaces are restricted to certain directions, while in a geometry without entanglement shadow the geodesic surfaces can pass point $x$ along any direction. As a consequence, the domain we obtained by our procedure (expanding minimal surfaces outwards and taking the overlap of their entanglement wedgs) leads to a small wedge with corners, rather than a disk. As an example, in Fig. \ref{fig:outshadow} we show the situation in a BTZ black hole. Starting from a geodesic that bounds half of the boundary (red solid curve), we move the geodesic while fix a point $x$ on it. Such motion ends at the red dashed line which bounds a different region that also has half the boundary size. The geodesic cannot move further beyond this point while still includes $x$. The intersection of small expansion of these geodesics give a diamond shape region (\ref{fig:outshadow} (b)) which means the bound on speed of light along the corner directions are loose. When point $x$ moves closer and closer to the entanglement shadow, the bound along the corner direction (which is the direction parallel to the black hole horizon) becomes worse and worse. When $x$ enters the entanglement shadow, local reconstruction does not give bound to bulk locality any more. % [I have modified the description here to make it less technical.]

\begin{figure}[htb!]
  \centering
  \includegraphics[width=0.9\textwidth]{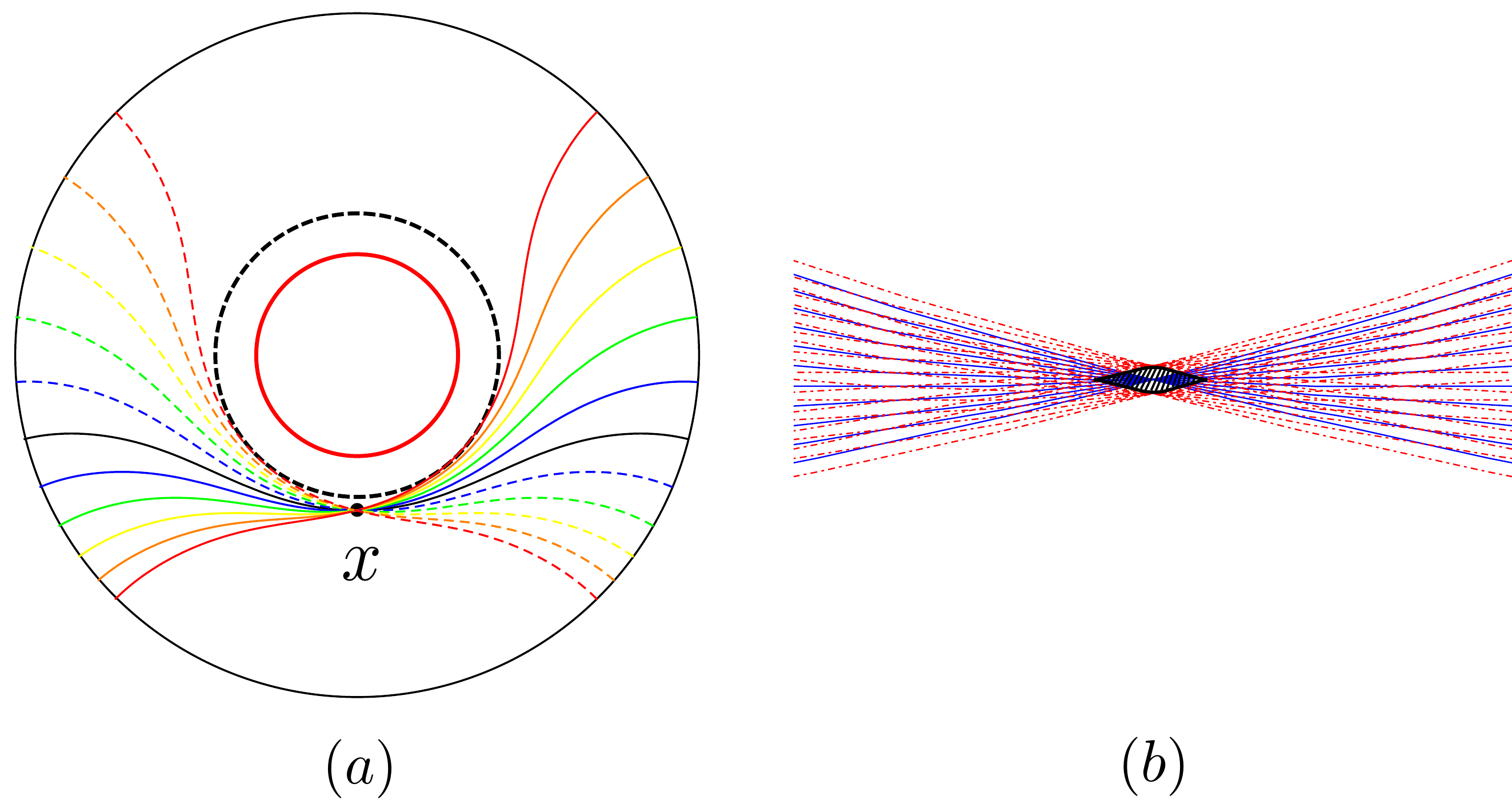}
  \caption{(a) The geodesics
    that cross the bulk point $x$ outside the entanglement shadow (the
    dashed black circle) in the BTZ blackhole geometry with the metric
    $ds^2 = -(r^2-b^2)dt^2 + (r^2-b^2)^{-1}dr^2+r^2 d\phi^2$. The red
    circle is the horizon $r=b$. We take $b=0.5$ in this calculation. We use coordinate
    $l=\frac{2}{\pi}\tan^{-1}(r)$ to map the infinite space
    $r\in[0,\infty]$to a finite disk $l\in[0,1]$. (b) A zoom-in picture of the intersection of expansions of geodesics crossing $x$. The overlap of the expanded geodesics gives the upper bound of causal future of $x$, which is the diamond shape region with corners, because the geodesics only cross $x$ from a finite angle range.}
  \label{fig:outshadow}
\end{figure}

\section{Bulk speed of light determines boundary butterfly velocities}\label{sec:protocol}
In this section, we discuss the other direction of the correspondence. Given a holographic theory with a known bulk space-time geometry, we would like to determine the butterfly velocity of a large family of boundary operators. %how to decide the butterfly velocity in the holographic given the bulk geometry. 
The reasoning is closely related to the discussion in the previous section. For a boundary operator $O_A$ supported in region $A$ which is a local reconstruction of bulk local operator $\phi_x$, we obtain an upper bound of the speed of light at point $x$ from the fact that $O_A$ cannot expand in space (of the boundary) faster than $v_{LR}$. Conversely, if we assume the exact speed of light at $x$, rather than its upper bound, `is already known, this generically imply a slower expansion speed of the dual operator $O_A$, which is its butterfly velocity. 

In contrast to the previous section, since we now assume the bulk
space-time geometry is given, we can work with a generally
time-dependent geometry in which the entanglement entropy of a
boundary region $A$ is given by the extremal surface area defined by
the Hubeny-Rangamani-Takayanagi (HRT) formula
\cite{hubeny2007covariant, wall2012maximin}. The entanglement
  wedge of a boundary region $A$ is a space-time region $E_A$ which is
  the domain of dependence for any space-like surface bounding $A$ and
  the extremal surface $\gamma_A$, as is illustrated in
  Fig.\ref{notation} (a). Any bulk operator in $E_A$ in the code
subspace can be locally reconstructed on
$A$.\cite{almheiri2014bulk,dong2016reconstruction,harlow2016ryu,dong2016bulk}

%We pick up the boundary code-subspace $\mathcal{H}_c$ such that all the boundary states in
%$\mathcal{H}_c$ are dual to the same classical bulk geometry.  Then we
%conclude that the causal structure in the bulk effective field theory
%decides the butterfly velocities of boundary operators. Since the bulk
%geometry is in its classical regime, we have implicitly assumed the
%large $N$ limit. 

%Based on the definition of the generic boundary operators and the
%quantum error correction properties in the holographic theories, we
%know that given a boundary region $A$, all the bulk operators
%$\phi_{E_A}$, that live in its entanglement wedge $E_A$ and touch the
%entanglement surface $\chi_A$, can be reconstructed as $O_A$, a
%generic boundary operator in the region $A$ (Fig.\ref{notation}
%(a))\cite{almheiri2014bulk,dong2016reconstruction,harlow2016ryu,dong2016bulk}.

%In this section, we study the butterfly velocities of the boundary
%operators that are dual to the local bulk operators  living on the
%entanglement surfaces $\chi_A$. We will demonstrate why the light cone
%of the effective theory in the bulk is related to the butterfly
%velocity on the boundary and explicitly build a protocol that derives
%the butterfly velocities from the bulk causal structure.

We start from a boundary spatial region $A$ at boundary time $t=0$, and a bulk point $x$ on the corresponding extremal surface $\gamma_A$ bounding $A$. (Again $x$ should be understood as actually an infinitesimal distance to $\gamma_A$ so that $x$ is included in $E_A$.) An operator $\phi_x$ at $x$ can be reconstructed in region $A$ as $O_{A}[\phi_x]$. We would like to determine the butterfly velocity $v(O_A[\phi_x];\mathcal{H}_c)$ of this operator in code subspace $\mathcal{H}_c$. (The code subspace is spanned by low energy excitations in the bulk that cause negligible back reaction.) 

%and set its boundary time as
%$t= 0$. Then we focus on a local bulk operator $\phi_x$ which lives on
%the extremal surface $\chi_{A}$. Because of the quantum error
%correction property, we reconstruct $\phi_x$ in the region $A$ as
%$O_{A}[\phi_x]$. Now we calculate the commutation relation between
%$O_A[\phi_x]$ and the boundary operators at time $\Delta t$ when
%$\Delta t\rightarrow 0$. 

We first state our conclusion. As has been defined in Eq. (\ref{eq:expansionofA}), we denote the expansion of region $A$ by a size $v\Delta t$ as $A_{v\Delta t}$. If operator $O_A$ grows with butterfly velocity $v$, it will spread to a region $A_{v\Delta t}$ at time $\Delta t$, as is illustrated in Fig.\ref{notation}(b). Denote $v^*$ as the minimal velocity such that the entanglement wedge $E_{A_{v\Delta t}}$ of $A_{v\Delta t}$ includes $x$. We claim that $v^*=v(O_A[\phi_x],\mathcal{H}_c)$ is the butterfly velocity of $O_A[\phi_x]$.

%One piece of notation we set here before the
%calculation is $A_{v\Delta t}$ (Fig.\ref{notation}(b)), which is the
%boundary region at time $\Delta t$ that satisfies
%\begin{equation*}
%  A_{v\Delta t} \equiv \{x\in \text{boundary}\left| \exists y\in A,  ~\text{s.t.}~  d(x,y) \leq  v\Delta t \right \}
%\end{equation*}
\begin{figure}[!htb]
  \centering
  \includegraphics[width=0.8\textwidth]{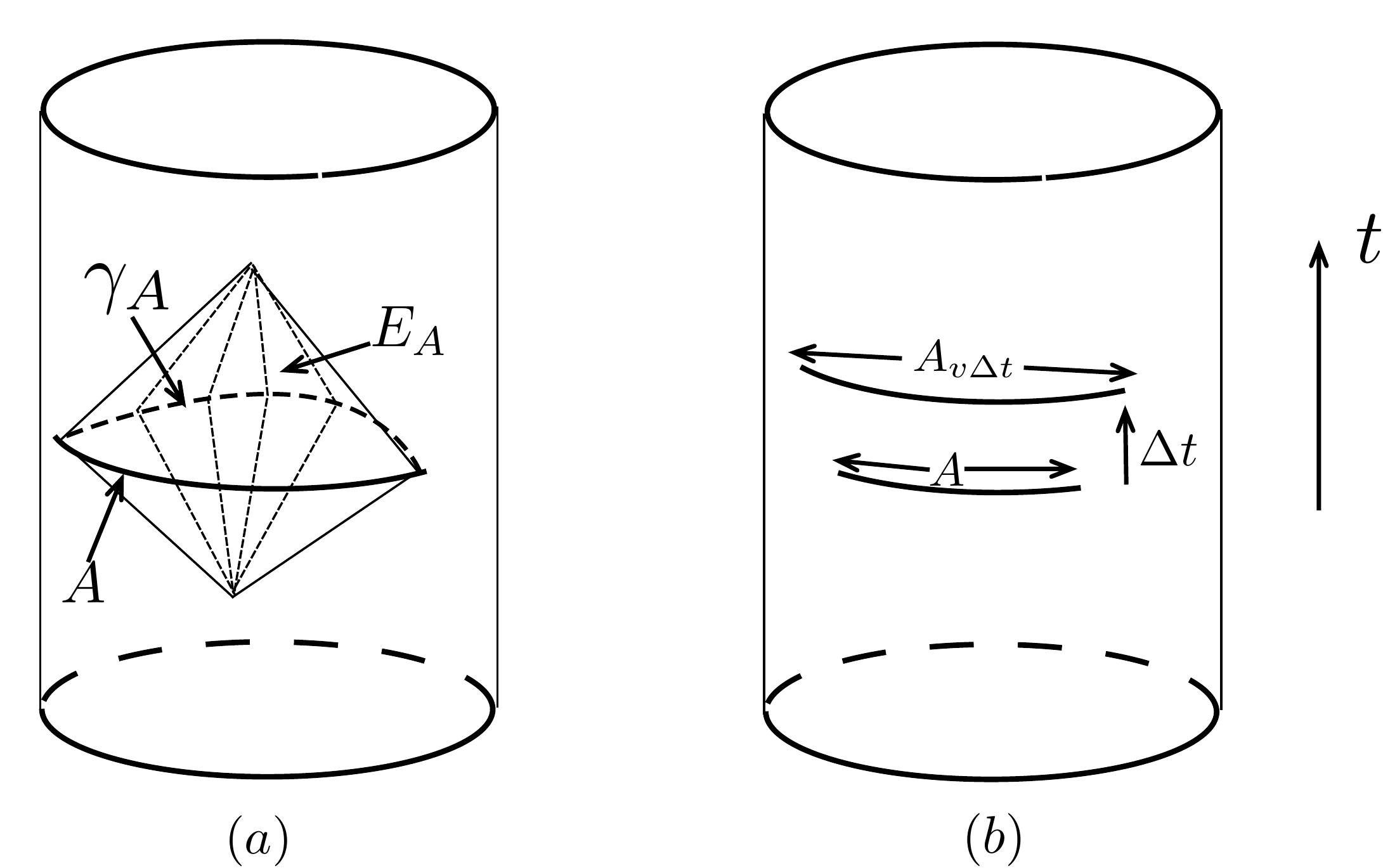}
  \caption{(a) Illustration of the entanglement wedge of a boundary region $A$. The vertical axis is the time direction and
    the boundary of the cylinder represents the boundary theory. The
    dashed curve $\gamma_{A}$ is the extremal surface bounding the region $A$. The entanglement wedge $E_A$ is the domain of dependence of the bulk region enclosed by $\gamma_A\cup A$. The thin straight lines are null geodesics. (b) The relation of region
    $A$ and its expansion $A_{v\Delta t}$ by speed $v$. }
  \label{notation}
\end{figure}
%Then we construct the boundary region $A_{v^*\Delta t}$ at time
%$t= \Delta t$, where $v^*$ is chosen such that the entanglement wedge
%of $A_{v^*\Delta t}$ just includes the bulk point $x$
%(Fig.\ref{protocol} (a)). Since the entanglement wedge of
%$A_{v^*\Delta t}$ include the bulk point $x$, $\phi_x$ can be
%reconstructed in the region $A_{v^*\Delta t}$ as
%$O_{A_{v^*\Delta t}}[\phi_x]$. We claim that the $v^*$ we pick up
%following this protocol is exactly the butterfly velocity,
%$v(O_{A}[\phi_x];\mathcal{H}_c)$, of the boundary operator
%$O_{A}[\phi_x]$.

\begin{figure}[ht!]
  \centering
  \includegraphics[width=0.8\textwidth]{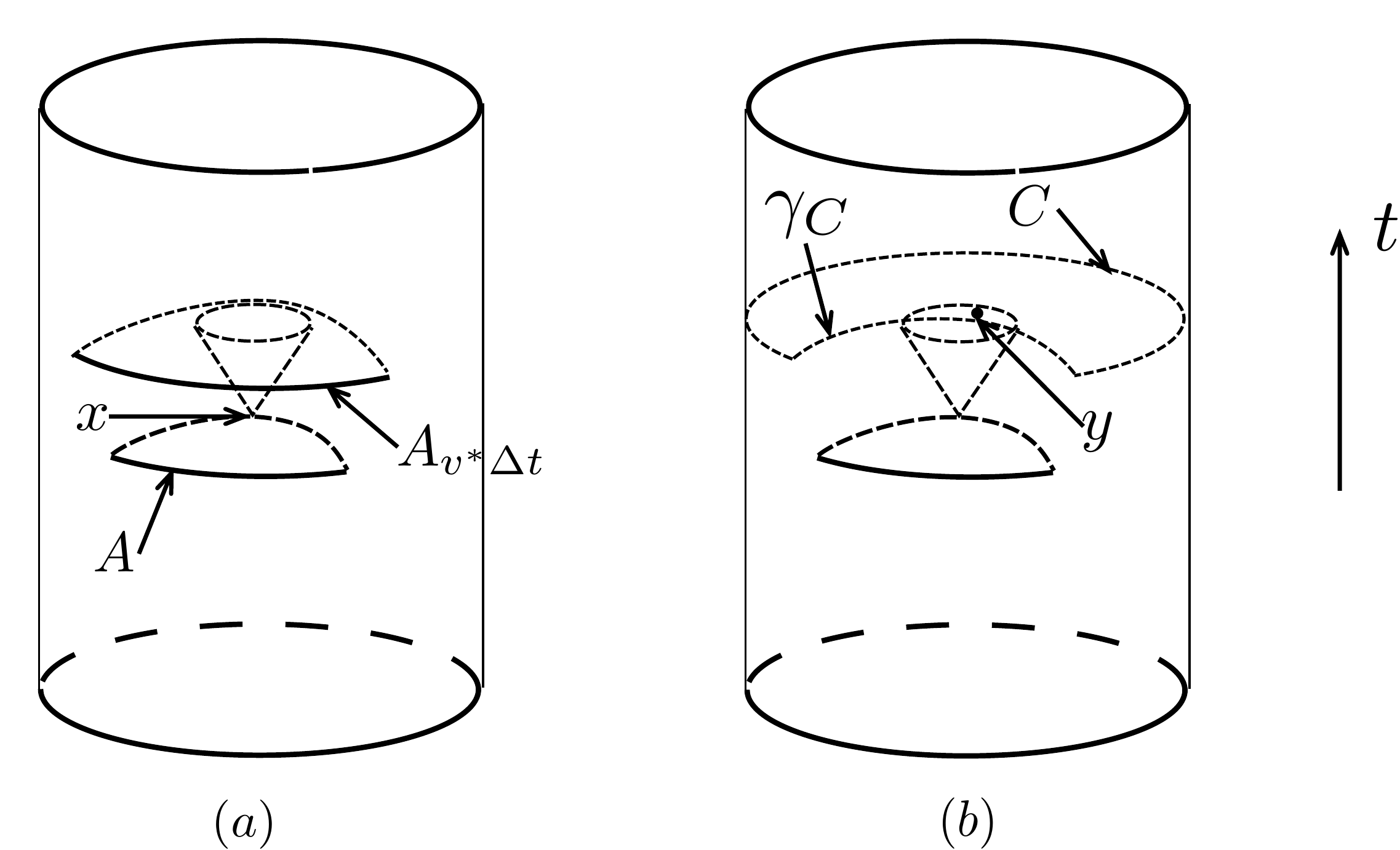}
  \caption{The setup that determines the butterfly velocity $v^*$. (a) For a region $A$ at boundary time $0$ and a point $x$ on the extremal surface bounding $A$, $v^*$ is chosen such that the extremal surface bounding the expanded region $A_{v^*\Delta t}$ at a later time $\Delta t$ is tangential to the light cone of $x$. In other words, $x$ is at the boundary of the entanglement wedge of $A_{v^*\Delta t}$. (b) An illustration why $v^*$ defined in (a) is the butterfly velocity. If we expand region $A$ by any velocity slower than $v^*$, the entanglement wedge of the complement region ($C$) will have a nontrivial intersection with the causal future of $x$. Therefore one can find operators in the intersection region which are reconstructed to $C$ and does not commute with $\phi_x$. This demonstrate that the boundary dual of $\phi_x$ really expand with speed $v^*$.} %For any distance $D<v^*\Delta t$, one can find a region $C$ with that distance to $A$ at boundary time $\Delta t$, with its entanglement wedge $E_C$ intersecting with the future of $x$. An operator $\phi_y$ with point $y$ in the intersection region can be reconstructed to $C$, which thus does not commute with $\phi_x$. Therefore the butterfly velocity
  
%  (a) illustrates how to find the boundary region
%    $A_{v^*\Delta t}$. $A$ is the boundary region at time $t=0$.
%    $\phi_x$ is located at the bulk point $x$, which lives on the
%    extremal surface $\chi_A$ of the region $A$. $A_{v^*\Delta t}$ is
%    the boundary region at time $t=\Delta t$, whose extremal surface
%    just enclose the future lightcone of $x$. In (b), $y$ is the bulk
%    point that is within the future lightcone of $x$ and is
%    arbitrarily close to the extremal surface $\chi_{A_{v^*\Delta
%        t}}$. Thus as long as $D<v^*\Delta t$, the region $C$ can be
%    found whose distance to the region $A$ is $D$. The bulk operator
%    $\phi_y$ can be reconstructed on $C$ as $O_c(\Delta t)$, which
%    does not commute with $O[\phi_x]$ in the code subspace.
  \label{protocol}
\end{figure}

To prove this conclusion, we need to prove two conditions. 1) Any operator on the complement of $A_{v^*\Delta t}$ (at time $\Delta t)$) commutes with $O_A[\phi_x]$. This proves the butterfly velocity $v(O_A[\phi_x])\leq v^*$. 2) For any velocity $v<v^*$, one can find operator on the complement of $A_{v\Delta t}$ that does not commute with $\phi_x$. This proves $v(O_A[\phi_x])= v^*$.

%$ \langle\psi_i| [O_A, B(\Delta t)]|\psi_j\rangle = 0$,
%$\forall ~ B\in \{R|d(R,A)= v(O_A;\mathcal{H}_c) \Delta t\}$,
%$\forall ~|\psi_i\rangle, |\psi_j\rangle \in \mathcal{H}_c$ and the we prove that $v^*$ is the minimum.  
We start from the first condition. By definition of $v^*$, $\phi_x$ can be reconstructed on the boundary as an operator $O_{A_{v^*\Delta t}}[\phi_x]$ in region $A_{v^*\Delta t}$ at time $\Delta t$. For any boundary operator $O_B(\Delta t)$ supported in a region $B$ at time $\Delta t$ that does not intersect $A_{v^*\Delta t}$,  we have
%respectively, thus $\forall O$ at time $\Delta t$, $\forall |\psi_i\rangle, |\psi_j\rangle \in \mathcal{H}_c$,
\begin{equation}
  \langle \psi_i | \big[O_{A}[\phi_x], O_B(\Delta t)\big]|\psi_j\rangle = \langle \psi_i | \big[\phi_x , O_B(\Delta t)\big]|\psi_j\rangle = \langle \psi_i | \big[O_{A_{v^*\Delta t}}[\phi_x], O_B(\Delta t)\big]|\psi_j\rangle =0
\end{equation}
for any pair of states $|\psi_i\rangle, |\psi_j\rangle \in\mathcal{H}_c$.  Therefore the butterfly velocity $v(O_A[\phi_x],\mathcal{H}_c)\leq v^*$. 

%Thus for the boundary region $B$ at the time slice $t =\Delta t$, as
%long as
%$B \cap A_{v^*\Delta t} = \emptyset \Rightarrow \big[O_{A_{v^*\Delta
%    t}}[\phi_x], O_B\big]=0$. From the definition of
%$A_{v^*\Delta t}$, the above statement is equivalent to that if the
%distance between the boundary region $A$ and $B$ satisfies
%$d(B,A)> v^* \Delta t \Rightarrow \langle \psi_i | \big[O_{A}[\phi_x],
%B(\Delta t)\big]|\psi_j\rangle = 0$.

Now we prove the second condition. By definition of $v^*$, for any $v<v^*$ the entanglement wedge 
$E_{A_{v\Delta t}}$ will not include $x$. Denote the complement of $A_{v\Delta t}$ as $C$, as is shown in Fig. \ref{protocol} (b). Since $x\notin E_{A_{v\Delta t}}$, the future of $x$ has a nontrivial intersection of $E_C$. In other words, there exists a bulk point $y\in E_C$ that is time-like separated to $x$. For a generic operator $\phi_x$, there exists operator $\phi_y$ which does not commute with $\phi_x$. Denote $O_C[\phi_y]$ as the reconstruction of $\phi_y$ on region $C$, we concluded that there exists some states $|\psi_i\rangle,~|\psi_j\rangle\in\mathcal{H}_c$, such that
\begin{eqnarray}
  \langle \psi_i | [\phi_x , \phi_y]|\psi_j\rangle \neq 0
  \Rightarrow  \langle \psi_i | \big[O_{A}[\phi_x] ,O_C[\phi_y]\big]|\psi_j\rangle \neq 0
\end{eqnarray}
This proves that the butterfly velocity must be larger than $v$ for any $v<v^*$. Therefore we reach the conclusion that $v^*$ is equal to the butterfly velocity $v(O_A[\phi_x],\mathcal{H}_c)$. 

%that is both arbitrarily close to
%the entanglement surface $\chi_{A_{v^*\Delta t}}$ and included in the
%future light-cone of point $x$ (Fig.\ref{protocol} (b)).  Since $y$ is
%arbitrarily close to the $\chi_{A_{v^*\Delta t}}$, then $\phi_y$ can
%be reconstructed on the boundary $C$ at time $t=\Delta t$, which has
%arbitrarily small intersection with $A_{v^*\Delta t}$. We denote the
%reconstruction of $\phi_y$ on the region $C$ as $O_C(\Delta t)$. On
%the other hand, because $\phi_x$ is the field operator of the
%effective field theory living on top of the semi-classical background
%geometry in the code-subspace,
%$\exists |\psi_i\rangle, |\psi_j\rangle \in \mathcal{H}_c$
%\begin{eqnarray*}
%  \langle \psi_i | [\phi_x , \phi_y]|\psi_j\rangle \neq 0
%  \Rightarrow  \langle \psi_i | \big[O_{A}[\phi_x] ,C(\Delta t)\big]|\psi_j\rangle \neq 0
%\end{eqnarray*}
%Because the region $C$ has arbitrary small intersection with
%$A_{v^*\Delta t}$, it means as long as $D< v^* \Delta t$, I can find
%the the region $C$ at time $t=\Delta t$, which does not commute with
%$O_{A}(\phi_x)$ in the code-subspace. Thus $v^*$ is the minimum velocity that satisfies Eq.\ref{def}.

We would like to emphasize that this protocol of determining the butterfly velocity is covariant, so that it applies to a generic geometry without time translation symmetry. In the following sections, we will study various properties of butterfly velocity based on this protocol. Although the duality is only established for asymptotic AdS geometries, it is well-defined to study the consequence of our protocol in even more general geometries, assuming the local reconstruction and HRT formula generalizes. For example, in Sec.\ref{sec:flat} we will apply this protocol to the flat space with a finite boundary, which provide conditions that the holographic dual theory of a flat space weakly coupled gravity have to satisfy, if such theory exists. 

%Also the protocol
%explicitly satisfies the diffeomorphism invariance in the bulk,
%because the construction of $A_{v^*\Delta t}$ only depends on the
%extremal surfaces and the light-cone of the bulk point $x$, which are
%independent of the coordinates in the bulk. Also, it means this
%protocol applies to the dynamical bulk geometry.

\section{General properties of butterfly velocity}\label{sec:generalproperties}

\subsection{Monotonicity}\label{sec:monobv}

In this section, we obtain some general properties of the butterfly
velocity based on the assumption that the bulk dual geometry satisfies
Einstein equation (EE) and the null energy condition (NEC).  In short,
we show that among bulk local operators $\phi_x$ on the same extremal
surface $\gamma_A$, the butterfly velocity is a non-increasing
function of the distance from $x$ to the boundary. Intuitively,
operators in the infrared always move with a velocity that is smaller
or equal to those in the ultraviolet. Our result applies to disk shape
regions on a boundary geometry with spatial rotation
symmetries. 

In the following we will explain our result and provide an intuitive explanation of the main idea of the proof. The rigorous proof will be given in Appendix. \ref{app:mono1}. We consider a disk shape region $A$ on the boundary. %$d-1$ dimensional spherical region $A$ on the boundary. 
When the boundary has rotation symmetry that preserves the disk $A$, the corresponding extremal surface $\gamma_A$ also has rotation symmetry. Points on $\gamma_A$ can be parametrized by $\Omega_{d-2}$, the $d-2$ dimensional angular coordinates, and $z$ which parameterizes the direction perpendicular to the boundary. The
boundary theory locates at $z=0$ (see Fig.\ref{mono1} (a)). Consider a bulk operator $\phi(z,\Omega_{d-2})$ which is locally reconstructed to region $A$ as a boundary operator $O_A[\phi(z,\Omega_{d-1})]$. Because of
the rotational symmetry, it is clear that the butterfly velocity $v\left(O_A[\phi(z,\Omega_{d-1})],\mathcal{H}_c\right)$ only depends on $z$, which we will denote as $v_A(z)$. Our result is that for any two points at $z_1<z_2$, $v_A(z_1)\geq v_A(z_2)$, as is illustrated in Fig.\ref{mono1} (b). 

%The
%statement is that, the deeper $\phi(z,\Omega_{d-1})$ goes into the
%bulk (the bigger the $z$ is), the smaller the butterfly velocity of
%its reconstruction in region $A$ will be ( Fig.\ref{mono1} (b)). This
%monotonicity result is consistent with our intuition that the
%butterfly velocities decrease when the operators move from the UV
%region to the IR region. We present the proof in
%Appendix.\ref{app:mono1}.

\begin{figure}[ht!]
  \centering
  \includegraphics[width=0.9\textwidth]{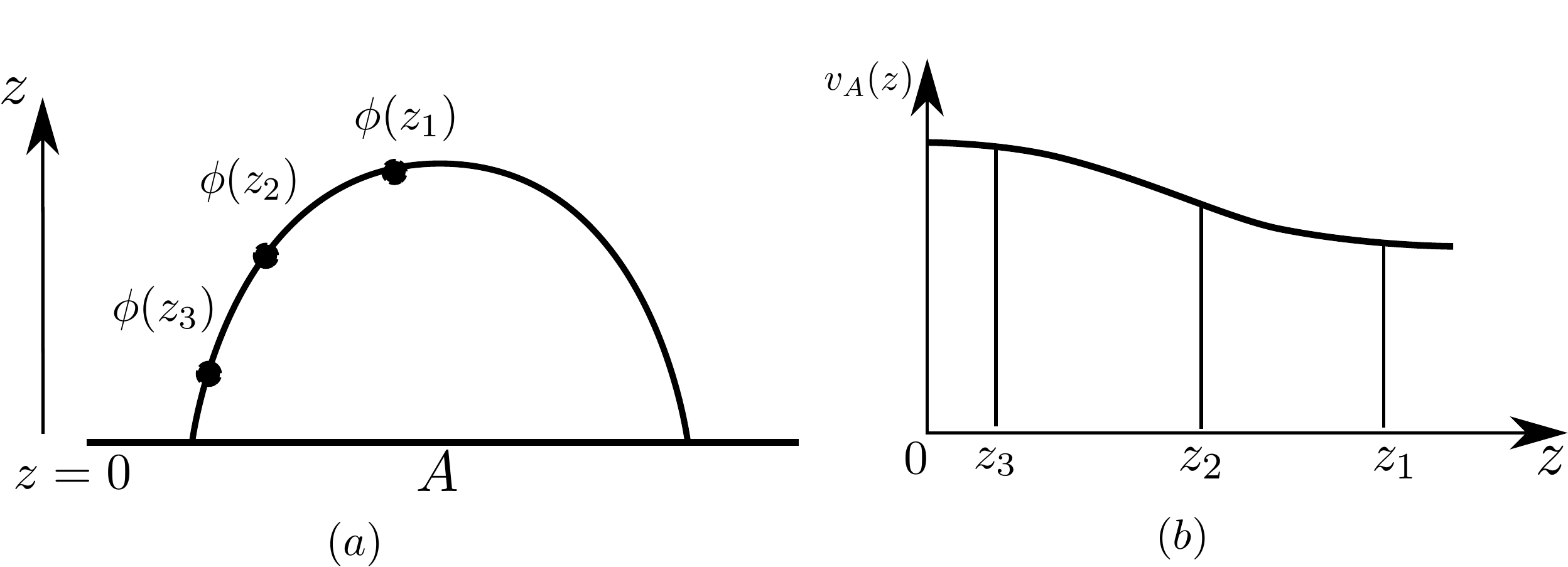}
  \caption{(a) Three bulk local operators $\phi(z_1), \phi(z_2), \phi(z_3)$ with decreasing distance to the boundary $z_1>z_2>z_3$, reconstructed to the same region $A$. (b) Schematic plot of butterfly velocity as a function of radial coordinate $z$ which decreases monotonously when the operator moves deeper in the bulk.}
  \label{mono1}
\end{figure}

This monotonicity property results from the property of null expansion in geometries satisfying EE and NEC, which has also played an essential role in proving the entanglement wedge is outside n the causal wedge in asymptotic AdS geometries\cite{wall2012maximin, hubeny2013global}. (To clarify, our result is not restricted to asymptotic AdS.) With more details reserved to Appendix \ref{app:mono1}, here we would like to provide some intuitive illustration to the proof for the simplest case of statistic geometries. We start by considering the butterfly velocity of the operator at the ``tip" of $\gamma_A$ (the point with $z=z_m$ maximal), which is determined by the minimal surface $\tilde{\gamma}$ in Fig. \ref{fig:proofsketch}. $\tilde{\gamma}$ is defined at boundary time $\Delta t$ and has a distance $c\Delta t$ to the tip point. In other words, as $\Delta t$ increases from $0$, the tip point of $\tilde{\gamma}$ grows with bulk speed of light. The butterfly velocity of the tip point $v_m$ is determined by the growth velocity of the boundary position of $\tilde{\gamma}$. As is shown in Fig. \ref{fig:proofsketch}, if $\tilde{\gamma}$ is anchored to a boundary disk that is bigger than $A$ by $x$, the butterfly velocity of tip operators is $v_m=x/\Delta t$. Now we pick a different point at depth $z<z_m$. The distance of point $z$ to the surface $\tilde{\gamma}$ is generically different from $c\Delta t$, which we denote as $u(z)\Delta t$. In other words, $u(z)$ is the speed of expansion of $\tilde{\gamma}$ at point $z$. To determine the butterfly velocity $v(z)$ of this point, one should consider another minimal surface that expands at $z$ point with speed of light, and expands at the boundary with velocity $v(z)$. Note that the expansion speed at different locations of the surface are proportional to each other, we have
\begin{eqnarray}
\frac{u(z)}{v_m}=\frac{c}{v(z)}\Rightarrow v(z)=\frac{c}{u(z)}v_m
\end{eqnarray}
Therefore the butterfly velocity of all points on $\gamma_A$ can be
determined by the single surface $\tilde{\gamma}$. To understand the
general behavior of $u(z)$, we draw some surfaces $C(d)$ which has
constant distance $d$ to the surface $\gamma_A$ (see
Fig. \ref{fig:proofsketch}). In space-time, such surfaces are obtained
by null expansion of $\gamma_A$ for time $d/c$. The key consequence of
EE and NEC is that any extremal surface like $\tilde{\gamma}$ cannot
be tangential to any $C(d)$ from outside \cite{wall2012maximin,
  hubeny2013global, headrick2014causality}. In other words, either
$\tilde{\gamma}$ coincide with $C(d)$ with $d=c\Delta t$ for all $z$,
or $\tilde{\gamma}$ crosses $C(d)$ surfaces, in which case the
distance $u(z)\Delta t$ has to decrease as $z$
decreases. Consequently, $v(z)\propto 1/u(z)$ always increases or
stays constant as $z$ decreases towards the boundary. A more rigorous
proof of this result is given in Appendix \ref{app:mono1}.

\begin{figure}[ht!]
  \centering
  \includegraphics[width=0.6\textwidth]{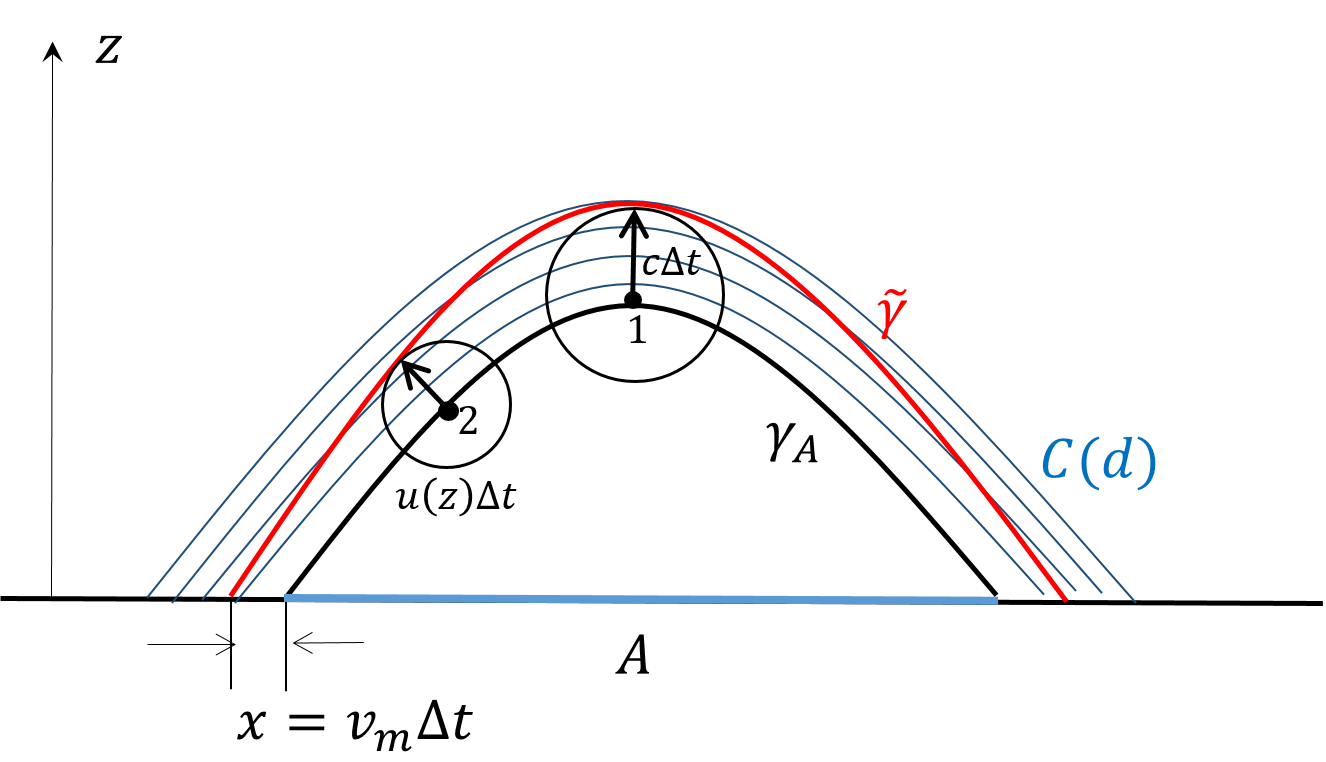}
  \caption{The setup in the proof of monotonicity of butterfly
    velocity (for the case of static geometries). The horizontal line
    at the bottom is the boundary at $z=0$. The black thick curve
    represents extremal surface $\gamma_A$ which bounds a boundary
    region $A$ (blue thick line). The red curve stands for a minimal
    surface $\tilde{\gamma}$ that has a distance $c\Delta t$ to the
    tip point $1$. At a different point $2$ with depth $z$,
    $\tilde{\gamma}$ expands with a different speed $u(z)$ so that it
    is $u(z)\Delta t$ away from point $2$. The blue think curves
    $C(d)$ are expansion of $\gamma_A$ by distance $d$ for different
    $d$. The key of the proof is that $\tilde{\gamma}$ cannot be
    tangential to $C(d)$'s anywhere other than the tip point. }
  \label{fig:proofsketch}
\end{figure}

One direct consequence of this result is that the butterfly velocity
of all boundary operators are smaller or equal to the speed of light
if the geometry is asymptotically AdS. A straightforward calculation
tells us that the butterfly velocity of all the operators are equal to
the speed of the light for a pure AdS geometry (see
Sec.\ref{example}). Thus for an asymptotic AdS geometry,
$\lim_{z\rightarrow 0}v_A(z) = c$. Consequently $v_A(z)\leq c$ for all
$z$ due to the monotonicity. This is consistent with our expectation
that the boundary theory with the asymptotic AdS dual is
relativistic. This is consistent with previously known results that there is no superluminal bulk signaling
between boundary points
\cite{gao2000theorems,engelhardt2016gravity,woolgar1994positivity}. In contrast, for geometries that are not asympotically
AdS, butterfly velocity can exceed speed of light even if the geometry
satisfies EE and NEC. One example is the flat space, which we will
study in Sec. \ref{sec:flat}.

%unfinished here

\subsection{When is the butterfly velocity equal to $c$?}\label{sec:causalwedge}

In the previous subsection, we have discussed that, for a boundary theory
dual to an asymptotic AdS gravity that satisfies EE and NEC, the
butterfly velocity of the boundary operators are less than or equal to the
speed of the light. It is natural to ask in general when the upper bound is saturated. In this subsection we will prove that the entanglement wedge and causal wedge of a boundary region $A$ coincide, then the butterfly velocity of all local operators on $\gamma_A$ is equal to the speed of light
$v(A;\mathcal{H}_c)=c$.

We will present the intuitive interpretation here and leave the
rigorous proof in the Appendix.\ref{app:proof}. In order to prove that
$v(A;\mathcal{H}_c)=c$ for all the generic operators when the causal
wedge of $A$ coincides with the entanglement wedge of $A$, we only
need to show that the entanglement wedge of $A_{u\Delta t}$ does not
contain any part of $\gamma_A$ if $u<c$. $\gamma_A$ is the
entanglement surface of $A$ which coincides with $A$'s causal
surface. The entanglement wedge $E_A$ is the domain of dependence of
the region enclosed by $A\cup \gamma_A$ (Fig.\ref{causalwedge}(a)). In
Fig.\ref{causalwedge}(b), the blue line is the extremal surface of the
boundary region $A_{c\Delta t}$ and the red line is the extremal
surface of the boundary region $A_{u\Delta t}$, where $u<c$. Because
the causal wedge of $A$ coincides with its entanglement wedge
$C_A=E_A$, then obviously $C_{A_{u\Delta t}}$, the causal wedge of
$A_{u\Delta t}$ (the red dashed line in Fig.\ref{causalwedge}(b)),
does not contain any part of $\gamma_A$ if $u<c$. On the other hand, in
the Appendix.\ref{app:proof}, we prove that, if $C_{A} = E_{A}$, then
for the boundary regions that are sufficiently close to $A$, such as
$A_{u\Delta t}$ and $A_{c\Delta t}$ with $\Delta t\rightarrow 0$, the
difference between the entanglement wedges and the causal wedges are
of $O(\Delta t^2)$. While the difference between $C_{A_{u\Delta t}}$
and $C_{A_{c\Delta t}}$ is of $O((c-u)\Delta t)$, thus as long as
$u<c$, $E_{A_{u\Delta t}}$ does not contain any part of
$\gamma_A$. Details are presented in Appendix.\ref{app:proof}.
% \XLQ{[Since we are giving intuitive explanation here, I feel it's
% better to not mention distance function. Just explain that the
% causal wedge and entanglement wedge only deviate from each other by
% $O(\Delta t^2)$ should be easier to understand. ]}

\begin{figure}[htb!]
  \centering
  \includegraphics[width=0.8\textwidth]{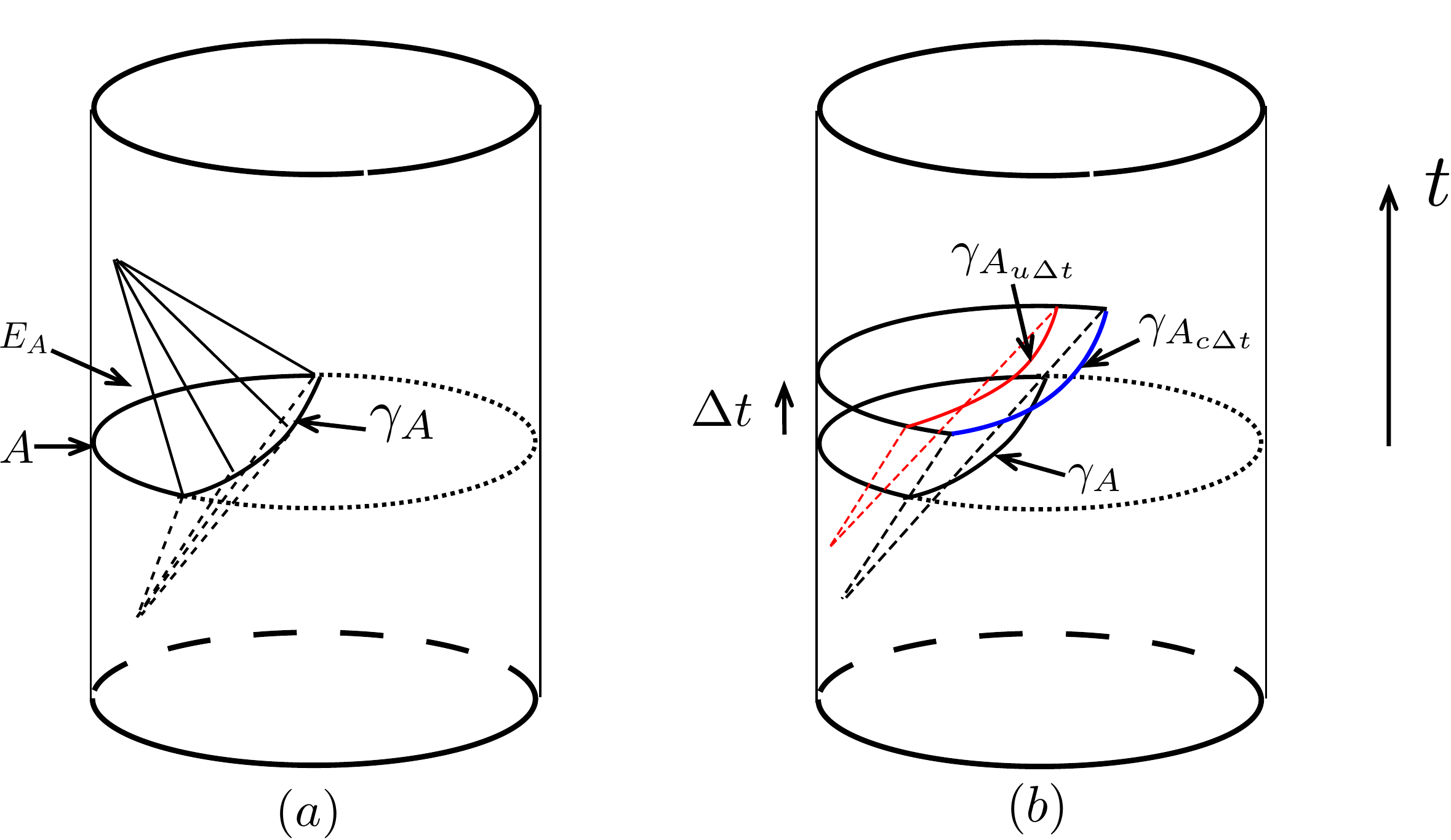}
  \caption{(a) Illustration of a region $A$ with its entanglement wedge coincide with causal wedge. (b) For any velocity $u<c$, the entanglement surface of the expansion $A_{u\Delta t}$ (red solid curve) does not intersect with $\chi_A$, so that the butterfly velocity $v^*>u$, which leads to the conclusion $v^*=c$.
%   In (a), the solid boundary region is $A$, the solid bulk
%    line is the entanglement surface $\chi_A$, which coincides with
%    $A$'s causal surface by assumption.  The entanglement wedge $E_A$
%    is the domain of dependence of the region enclosed by
%    $A\cup \chi_A$. In (b), the red and the blue solid line are
%    $\chi_{A_{u\Delta t}}$, the entanglement surface of
%    $A_{u\Delta t}$, and $\chi_{A_{c\Delta t}}$, the entanglement
%    surface of $A_{c\Delta t}$ respectively.  The red dashed line is
%    the past causal wedge of $A_{u\Delta t}$. It is obvious that
%    $C_{A_{u\Delta t}}$, the causal wedge of $A_{u\Delta t}$ does not
%    contain any part of $\chi_A$ if $u<c$.
 }
  \label{causalwedge}
\end{figure}

\section{Examples}\label{example}
\subsection{Pure AdS and BTZ black hole}
In this subsection, we follow our protocol in Sec.\ref{sec:protocol}, and
calculate the butterfly velocities of operators in boundary theory
whose bulk dual geometry is the $d+1$ dimensional AdS space. A
straightforward calculation shows us that the butterfly velocities are
indeed the speed of light $c=1$. This consistent with our conclusion
in Sec.\ref{sec:causalwedge}, since in pure AdS the causal wedges
coincide with the entanglement wedges.

The metric of $d+1$ dimensional AdS is
$ds^2 = (-dt^2 + dz^2 + \sum_i dy_i^2)/z^2 $. Since the space-time is
translationally invariant in $t$ and $y_i$, without loss of
generality, we focus on the spherical boundary region $A$ at
$t=0$ centered at $y_i = 0$. The radius of the boundary region is $R$,
and the minimal surface $\gamma$ that covers this boundary region is
$ \sum_i y_i^2 + z^2 = R^2$. For an arbitrary point $x$ on the minimal
surface, we can choose the coordinate so that $y_i(x) = 0$,
$i=2\cdots d-1$, thus $y_1(x)^2+z(x)^2 = R^2$.

We then expand the radius of the boundary region a little bit to
$R+\Delta R$ and the new minimal surface $\gamma^\prime$ is
$\sum_i (y_i(x)+\Delta y_i)^2 +(z(x) +\Delta z)^2 = (R+\Delta R)^2$. $\chi^\prime$ is located at time slice $t=\Delta t$. A straightforward calculation shows that the shortest distance between point $x$ and $\gamma^\prime$ is 
\begin{equation}
  \label{eq:5}
  d(x,\gamma^\prime) =  \frac{1 }{z(x)^2} \min\left(\sum_i\Delta y_i^2 +\Delta z^2 -\Delta t^2\right) = \frac{1 }{z(x)^2}\left( -\Delta t^2 + \Delta R^2 \right)
\end{equation}
The butterfly velocity is determined by choosing a $\gamma^\prime$ that is lightlike separated from $x$, with $d(x,\gamma')=0$. This requires $\Delta t=\Delta R$ and gives the butterfly velocity 
%Thus in order that the light-cone of point $x$ is just enclosed by
%$\chi^\prime$, $d(x,\chi^\prime) = 0$, which means
%$\Delta t = \Delta R$. Thus the butterfly velocity of the bulk
%operator $\phi(x)$ that is reconstructed on the boundary region $A$
\begin{equation}
  \label{eq:6}
  v(O_A[\phi_x];\mathcal{H}_c) = \frac{\Delta R}{\Delta t} = 1
\end{equation}

Since the $2+1$ BTZ blackhole is locally identical to pure AdS space, the butterfly velocity for the dual of bulk local operators outside the entanglement shadow of BTZ black hole is also equal to speed of light. 
%since $2+1$ AdS space is the universal covering space of $2+1$ BTZ black hole. Thus we conclude that outside the entanglement shadow in $2+1$ BTZ, the butterfly velocity of the bulk operators after being reconstructed onto the boundary is $c=1$. 
This is consistent with our
understandings that the butterfly velocities of operators in
$1+1$d CFT at the finite temperature is still $c$.

\subsection{Higher dimensionsal AdS Schwarzschild black hole} \label{example2}

In higher dimensions, butterfly velocities of operators in finite temperature
systems are smaller than the speed of the light. The butterfly
velocities of operators evolved after the scrambling time has been
calculated in both the shock wave geometry approach
\cite{shenker2013black, roberts2016lieb, roberts2014localized} and the near-horizon minimal surface approach
\cite{mezei2016entanglement}. Our result can be considered as a generalization of the latter. 

In this subsection, we systematically study the butterfly velocities of
operators in a finite temperature CFT in $d$ dimensions, for which the
bulk dual geometry is the $d+1$-dimensional AdS Schwarzschild blackhole.

The metric of AdS Schwarzschild blackhole is
\begin{eqnarray}
  ds^2 = - U(r) dt^2 +U(r)^{-1} dr^2 +r^2 d\Omega_{d-1}^2 ~~~~~~
  U(r) = 1-\frac{\mu}{r^{d-2}} + \frac{r^2}{l^2} 
\end{eqnarray}
with $l$ the AdS radius. $\mu$ is related to the black hole mass by
$M = \Omega_{d-1} \mu (d-1) /(16\pi G_N)$, where $\Omega_{d-1}$ is the
area of a $d-1$-dimensional unit radius sphere. For simplicity, we fix $l=1$. In
this coordinate, $r\rightarrow\infty$ is the boundary. For $d>2$, the
extremal surfaces do not have a nice analytic form, thus we
numerically implement our protocol and calculate the butterfly
velocities of generic operators supported on spherical boundary
regions.

For rotation symmetric boundary region, the minimal surface is also rotation symmetric. The intersection of the minimal surface at a fixed radial coordinate $r$ is a spherical cap on the sphere $S^{d-1}$. We can choose a coordinate $d\Omega_{d-1}^2 = d\theta^2 + \sin^2\theta d\Omega_{d-2}^2$ for the fixed $r$ sphere, with $\theta=0$ at the center of the sphere. The minimal surface is parameterized by a function $r(\theta)$. 
%Because we only care about the spherical regions on the boundary, we expand $d\Omega_{d-1}^2 = d\theta^2 + \sin^2\theta d\Omega_{d-2}^2$, and the extremal surfaces can be parametrized as $r(\theta)$. 
Besides, the geometry of AdS Schwarzschild black hole is static, so that the
extremal surfaces live on a constant $t$ slice.  The area of
this co-dimension $2$ surface is
\begin{equation}
  Area = \Omega_{d-2}  (r\sin\theta)^{d-2} \sqrt{U(r)^{-1} \left(\frac{dr}{d\theta}\right)^2  +r^2} 
\end{equation}
%where in this expression $r$ is a function depends on $\theta$.
$r(\theta)$ is determined by minimizing the area.

We specify a boundary region $A$ by the maximal depth that
$A$'s extremal surface penetrates into the bulk, given by
$r_A(\theta=0) = r_0$, and $\frac{dr_A}{d\theta}|_{\theta=0}= 0$.
Then we find the solution $r_A(\theta)$ with this boundary condition, and another solution
$\tilde{r}_A(\theta)$ with a slightly different boundary condition
$\tilde{r}_A(0) = r_0+ \delta$,
$\frac{d\tilde{r}'_A}{d\theta}|_{\theta=0}= 0$, with $\delta \ll r_0$.
From these two solutions, we can decide numerically the butterfly
velocities of the boundary reconstruction of the bulk operator
$\phi(r_x)$, which is located at $r= r_x$, $\theta_x=r_A^{-1}(r_x)$.
\begin{eqnarray}
  \label{eq:2}
  &&v\left(O_A[\phi(r_x)];\mathcal{H}_c\right) =\left. \frac{\Delta \theta}{\Delta t}\right|_{r\rightarrow \infty}  \\
\nonumber &&  \left. \Delta \theta\right|_{r\rightarrow \infty} = \tilde{r}_A^{-1}(\infty) - r_A^{-1}(\infty) \\
\nonumber &&  \left. \Delta t \right|_{r\rightarrow \infty} = \sqrt{ \frac{  \min_\theta \left[ U(r_x)^{-1} (r_x - \tilde{r}_A(\theta))^2 + r_x^2 \left(\theta_x -\theta \right)^2 \right]   }{ U(r_x)} }
\end{eqnarray}

We plot the butterfly velocities
$v\left(O_A[\phi(r_x)];\mathcal{H}_c\right)$ of the boundary
reconstructions of the bulk operators $\phi(r_x)$ living on the
background of $3+1$, $4+1$, $5+1$ AdS Schwarzschild black hole
(Fig.\ref{fig:AdSBH}). The black horizontal lines are $\sqrt{\frac{d}{2(d-1)}}$,
the butterfly velocity predicted by \cite{shenker2013black,
  roberts2016lieb, roberts2014localized, mezei2016entanglement} of the
operators evolved after the scrambling time.  $d$ is the spatial
dimension of the bulk geometry. Each curve corresponds to the butterfly velocity of points on a minimal surface with fixed $r_0$. 

%As for the red, the orange and the
%blue lines, if we follow a single colored curve, it shows the
%evolution of the the butterfly velocities when we reconstruct the bulk
%operators at different depth $r_x$ on the extremal surface onto the
%same boundary region. The red and orange lines correspond to the
%boundary regions whose extremal surfaces stop penetrating into the
%bulk at $2r_{BH}$ and $(1+10^{-2})r_{BH}$, while the blue lines in
%(a), (b), (c) correspond to the boundary region whose extremal
%surfaces stop at $(1+10^{-4})r_{BH}$, $(1+4*10^{-5})r_{BH}$,
%$(1+8*10^{-6})r_{BH}$. $r_{BH}$ is the location of the horizon.
\begin{figure}[ht!]
  \centering
  \includegraphics[width=0.9\textwidth]{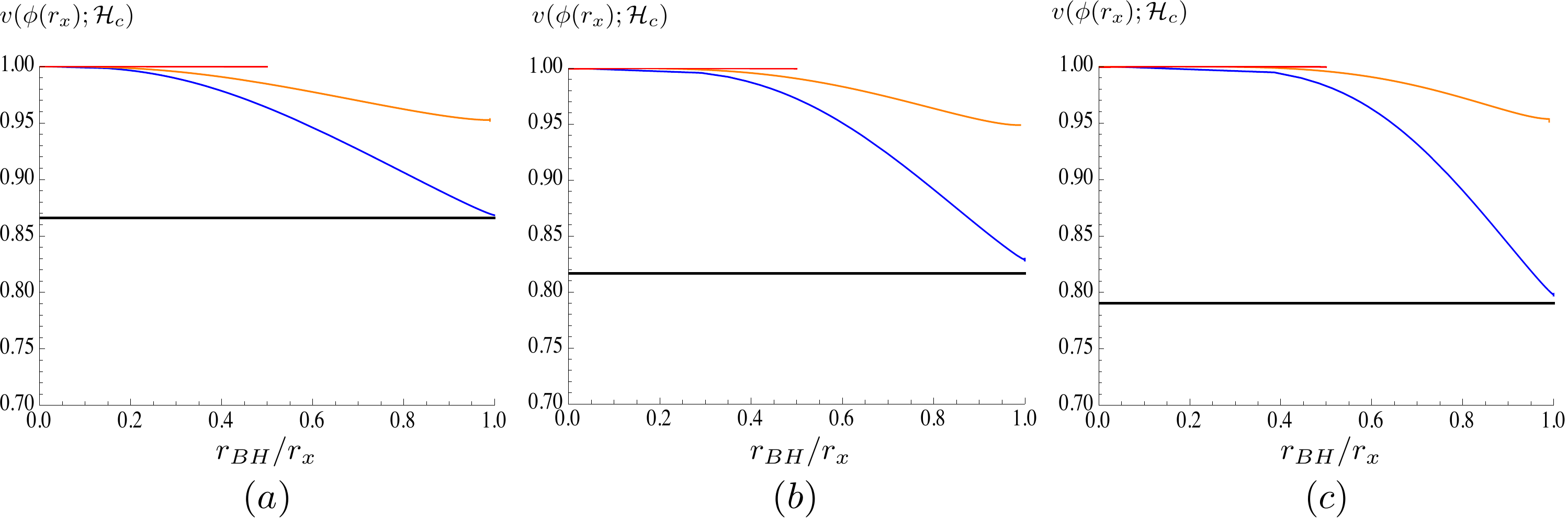}
  \caption{(a), (b), (c) are plots of the butterfly velocities
    $v(\phi(r_x);\mathcal{H}_c)$ of the boundary reconstructions of
    the bulk operators $\phi(r_x)$ in $3+1$, $4+1$, $5+1$-dimensional AdS
    Schwarzschild black hole, respectively. The black lines is the reference value $\sqrt{d/2(d-1)}$,
    with $d$ the bulk spatial dimension. %As for the red, the orange and the blue lines, if we follow a single colored curve, we vary the depth $r_x$ of the bulk operator $\phi(r_x)$, but we reconstruct $\phi(r_x)$ onto the same boundary region and plot the butterfly velocity of the reconstructed operators.  
    Each curve corresponds to the butterfly velocity of operators on a minimal surface as a function of their radial coordinates. The red and orange curves correspond to minimal surfaces with tip at radial coordinate $r_0=2r_{BH},~(1+10^{-2})r_{BH}$, respectively. The blue lines in (a),
    (b), (c) correspond to $r_0=(1+10^{-4})r_{BH}$, $(1+4\times 10^{-5})r_{BH}$,
    $(1+8\times 10^{-6})r_{BH}$, respectively. Here $r_{BH}$ is the location of the horizon.  }
  \label{fig:AdSBH}
\end{figure}

In Fig.\ref{fig:AdSBH}, we confirmed our conclusion in Sec. \ref{sec:monobv} that for spherical boundary
regions, the butterfly velocities of the reconstructed operators
decrease monotonically when its corresponding bulk operator moves to the IR. We also notice that the butterfly velocity approaches the universal IR value $\sqrt{\frac{d}{2(d-1)}}$ only when the bulk operators are
extremely close to the horizon.

It is interesting to note that due to translation symmetry, the points on different curves in Fig.\ref{fig:AdSBH} with the same $r$ coordinate can be considered as different boundary local reconstructions of the same bulk operator (Fig.\ref{mono2}). Although these different reconstruction operators all act identically in the code subspace, they act differently outside the code subspace, and have different butterfly velocities. Our numerical result indicates that the reconstructed operator in a bigger region (bounded by a minimal surface with smaller $r_0$) always has a smaller butterfly velocity. It is interesting to ask whether there is any monotonicity of butterfly velocity as a function of operator size. For general geometries this is clearly not true, since we can consider a geometry which is vacuum in IR and has matter in UV which are falling in. For certain local operators in the bulk, we can find local reconstructions in a big region that is completely in the vacuum, while smaller reconstructions have to enter the region with matter, so that the butterfly velocity is maximal for the former. It is an interesting question whether the monotonicity in operator size is correct for a restricted class of geometries, such as static geometries with translation and rotation symmetry. It is still interesting if that is true. We leave this as an open question for future works. 

%Besides, if we look at the bulk operators that locate at the same
%depth in the bulk, our numerical results indicate that among all the
%boundary reconstructions of the same bulk operator, the butterfly
%velocities of the reconstructed operators decrease monotonically when
%the size of the supported boundary region increases. Thus we make the
%following conjecture. From the quantum error correction, a bulk local
%operator $\phi_x$ can be reconstructed from different boundary
%regions. We only care about the reconstruction of the bulk operator as
%the generic operator from the spherical boundary regions $A(R)$ whose
%radius is $R$.  Due to the translational and the rotational symmetry,
%the butterfly velocities of the reconstructed operators are only
%decided by the radius $R$ of the boundary region (Fig.\ref{mono2}(a))
%and it decreases monotonically with respect to the size of the
%boundary reconstructing region $R$ if the homology condition of the
%extremal surfaces does not change (Fig.\ref{mono2}(b)). 
\begin{figure}[ht!]
  \centering
  \includegraphics[width=0.7\textwidth]{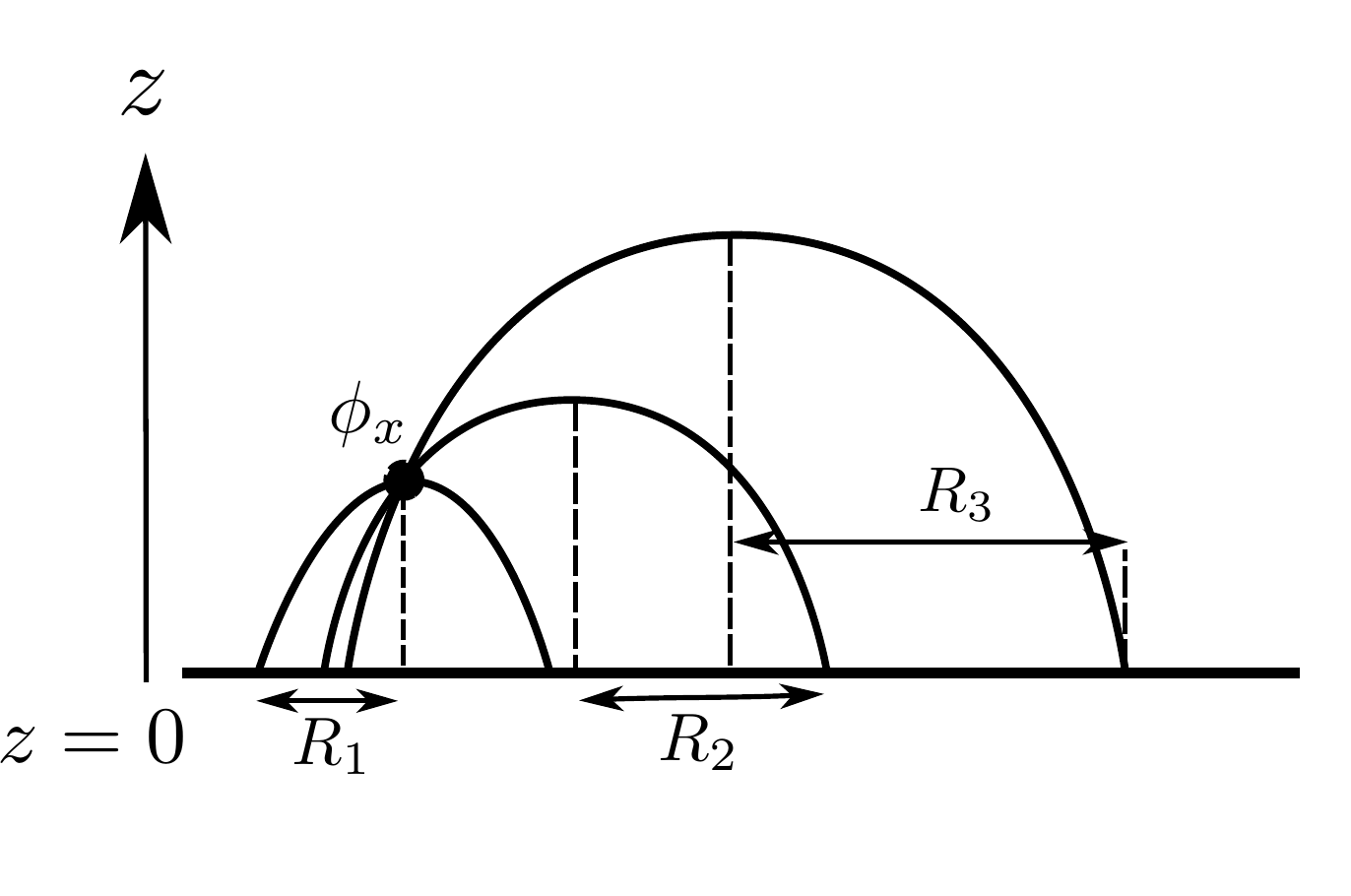}
  \caption{The same operator $\phi_x$ locally reconstructed to boundary regions $R_1$, $R_2$, $R_3$ with different size. Whether there is a monotonicity of the butterfly velocity vs size of the region is an open question. 
  }
  \label{mono2}
\end{figure}

\section{Flat space}\label{sec:flat}
Since our general framework makes no assumptions on the geometry, it
is natural to generalize our discussion to bulk geometries that are
not asymptotic AdS. Although holographic duality has not been
generalized to other geometries, the concepts such as HRT surface,
entanglement wedge and causal wedge, etc., are well-defined for any
Riemann geometry with a boundary. \footnote{In asymptotic AdS geometry
  the boundary is conformal boundary, while in general geometry we may
  consider a boundary at a finite location. For a finite boundary the
  gravity is not decoupled from the boundary theory, but this does not
  affect our discussion here in the large $N$ limit, since concepts
  like local reconstruction and entanglement wedge are all properties
  of the classical background geometry. } Therefore we can ask the
following question: If there is a boundary theory which is dual to a
given bulk geometry, in the sense that local reconstruction properties
apply to bulk low energy operators in the same way as the asymptotic
AdS case, how will this boundary theory look like?  The HRT formula
for FRW geometries with a spherical boundary has been studied in
Ref. \cite{nomura2016spacetime, nomura2016toward}, which shows that
the entanglement entropy of the boundary follows a volume law. In the
following, we will study this problem for the flat space with a
spherical boundary, from the point of view of butterfly
velocities. Following our protocol in Sec. \ref{sec:protocol}, we will
find conditions on butterfly velocities that have to be satisfied for
any possible dual of flat space gravity. In particular, we found that
the butterfly velocity is not bounded from above, which indicates that
the dual theory of flat space gravity, if exists, has to be
nonlocal. This is consistent with the high entanglement entropy found
in Ref.\cite{nomura2016spacetime, nomura2016toward}.

We consider the $d+1$-dimensional flat space with the metric
\begin{equation}
  \label{eq:1}
  ds^2 = -dt^2 + dr^2 + r^2 d\theta^2 +r^2\sin\theta^2 d\Omega_{d-2}^2 
\end{equation}
and a spherical boundary at $r=\Lambda$. The induced
metric of the boundary is
$ds^2 = -dt^2 + \Lambda^2d\theta^2 +\Lambda^2\sin\theta^2
d\Omega_{d-2}^2$.   

We focus on the butterfly velocity of the reconstruction of bulk local operators on disk regions. % boundary operators living on the spherical regions, thus we characterize the boundary region by the
A disk region centered at $\theta=0$ point is defined by $r=\Lambda$,
$\theta\in[0,\Theta]$. The disk region is a cap on the boundary sphere, and the minimal surface is a flat disk bounding the cap, as is illustrated in Fig. \ref{fig:flat}. The calculation of butterfly velocity is straightforward. Here we will list the main results and leave more detail of the explicit calculation in
Appendix.\ref{app:flat}. 

\begin{figure}[ht!]
  \centering
  \includegraphics[width=0.4\textwidth]{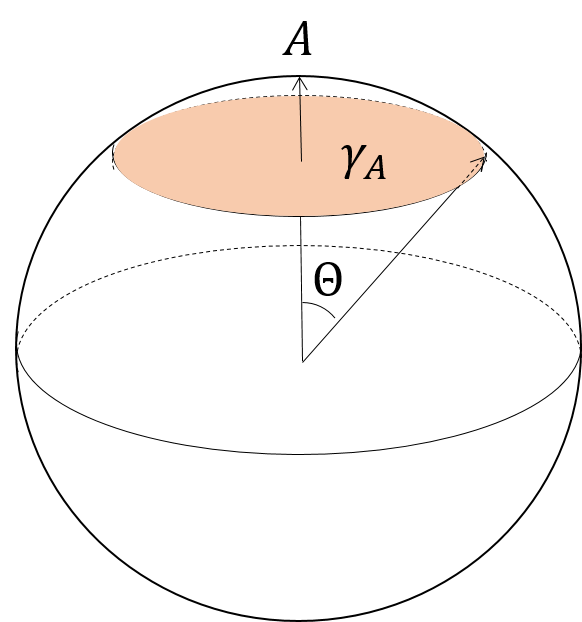}
  \caption{Illustration of the flat space with a spherical boundary. A boundary spherical cap region $A$ parameterized by angle coordinate $\Theta$ is bounded by a flat minimal surface $\gamma_A$.}
  \label{fig:flat}
\end{figure}

\begin{enumerate}
\item For operators $\phi_x$ at different location of the same minimal surface, their reconstruction $O_A[\phi_x]$ on the same boundary region all have the same butterfly velocity. This can be seen easily by applying the intuitive picture we discussed in Sec. \ref{sec:monobv}. If we expand the boundary region by increasing $\Theta$, the minimal surface stays flat so that it expands in the bulk with a constant velocity. 
\item Due to the property $1$ discussed above, the butterfly velocity $v_B[\phi_x]=v_B(\Theta)$ is only a function of the reconstructed operator size $\Theta$. $v_B(\Theta)$ has the following form:
  \begin{equation}
    \label{eq:flat}
    v(O_{A(\Theta)};\mathcal{H}_c) = \frac{c}{\sin\Theta}  ~~~ 0\leq\Theta\leq \pi
  \end{equation}
  Therefore the butterfly velocity of 
boundary operators diverge in the small size limit $\Theta\rightarrow 0$. The bulk speed of light is actually the lower bound of $v_B$, which is reached by biggest operators that occupy half of the boundary.  %, so will be the Lieb-Robinson velocity. It means the effect of the local operations spread infinitely fast.  Thus the boundary theory is non-Lorentzian.
\end{enumerate}

Therefore we have shown that the dual of flat space gravity, if
exists, must have a divergent Lieb-Robinson velocity, which requires
the Hamiltonian to be nonlocal. Conversely, we can also apply the
discussion in Sec. \ref{sec:TN} to the flat space case. If we assume
the boundary theory is a Lorentz invariant theory and is mapped to a
bulk theory on flat space (say by a tensor network), we conclude that
the bulk light cone has to have a strange shape. The speed of light in
the direction perpendicular to the boundary for a bulk point $x$ has
to vanish as the point approaches the boundary. It is interesting to
generalize this discussion to more generic geometries such as the FRW
geometry studied in Ref.\cite{nomura2016spacetime, nomura2016toward}.

\section{Conclusion and discussion}\label{sec:conclusion}

In conclusion, in this paper we provide a general definition of butterfly velocity, which characterizes the propagation velocity
of an operator measured in a given code-subspace. For large $N$ theories with a gravity dual, we show that the quantum error correction properties in local reconstruction of bulk operators closely relates the bulk causal structure and the boundary butterfly velocities. 

%Furthermore, we prove that, in the holographic theory, the butterfly velocity is related to the bulk causal structure via the quantum error correction condition.

This relation is bidirectional. In the direction from the boundary to
the bulk, we show that the Lieb-Robinson velocity of the boundary theory
constrains the location of the bulk light cone, which guarantees that a local boundary theory is mapped to a bulk theory that appears local in the code subspace. In the
direction from the bulk to the boundary, the bulk speed of light and extremal surfaces determine the butterfly velocity of boundary operators which are dual to bulk local operators. This correspondence has many consequences. For a spherical region in rotation invariant geometries, with the condition of EE and NEC we prove that the
butterfly velocity of the dual of a bulk local operator decreases monotonously as the bulk operator moves from UV to IR. When the causal wedge of the boundary region
$A$ coincides with its entanglement wedge in an asymptotic AdS
geometry (satisfying EE and NEC), the butterfly velocities of generic operators on $A$ are
exactly the speed of light. Explicit examples, including pure AdS and
AdS Schwarzschild black holes in different dimensions, are studied, which confirm our new results and is consistent with the previous results
\cite{shenker2013black, roberts2014localized, roberts2016lieb,
  mezei2016entanglement} in suitable limits. We have also
applied our result beyond the standard AdS/CFT and obtain constraints on a possible dual theory of flat space gravity. We observe that in this case the
boundary theory has to be nonlocal, with a diverging butterfly velocity for local operators.

% \XLQ{[I put reference to 29-31 to the end of Sec. 5.1, and add a
% short reference to the rest in the end of conclusion. I think this
% paragraph can be removed.] }It is worth noting that many other
% attempts have been made to understand the relation between the bulk
% dynamics and the boundary dynamics in AdS/CFT correspondence. For
% example, it has been proved that, under certain necessary
% assumptions, there is no superluminal bulk signaling between
% boundary points
% \cite{gao2000theorems,engelhardt2016gravity,woolgar1994positivity}. \ZY{Authors
% of \cite{engelhardt2016towards, engelhardt2016recovering} aim to
% reconstruct the bulk metric from the boundary correlation functions
% on the ``light-cone cut''. The ``light-cone cut'' is the spatial
% cross-sections of the boundary at infinity which is on the
% light-cone of a bulk point.} \XLQ{[Insert a discussion here about
% Gary Horowitz's work on bulk causal structure.]}  \ZY{[Done.]} Also
% in the framework of multiscale entanglement renormalization ansatz
% (MERA) \cite{vidal2007entanglement}, \XLQ{[is their result more
% about MERA or EHM? give references to MERA/EHM here. Is it a
% conjecture or argument/proof? ]} \ZY{[It is mainly about MERA]}, it
% was \ZY{argued} that the light cone in the bulk is related to the
% ``Lieb-Robinson type'' velocity on the
% boundary\cite{kim2016entanglement}.

There are many open questions that shall be studied in future works. In the black hole geometry we observe that different reconstruction of the same bulk operator has different butterfly velocitites, and the reconstruction on a smaller region corresponds to a faster butterfly velocity. It is natural to search for more general constraints on the butterfly velocity as a function of the size of the boundary operator. Such relation may exist in general boundary theories, or may provide further conditions that a holographic theory has to satisfy.

%We
%conjecture in Sec.\ref{example2}, that among all the reconstructions
%onto the spherical boundary regions of the same local bulk operator,
%the bigger the boundary region is, the slower the butterfly
%velocity. Although numerical evidences have been provided in this
%paper. Further analytical proof are also important.

Our discussion has focused on boundary operators that are local reconstruction of bulk local operators. If we generalize the discussion to more generic boundary operators that are dual to multi-point operators in the bulk, can we may obtain more general relation between the butterfly velocity of the boundary operator and the location of the dual bulk operator. For example, in the black hole geometry in $d+1>3$, it is reasonable to believe there is only one boundary operator with the slowest butterfly velocity, which is dual to the local operator at the tip. If a given boundary operator has a butterfly velocity that is above the minimal value but smaller than speed of light, we know it cannot contain the faster operators living on the UV part of the minimal surface. In general, for a given boundary operator, its butterfly velocity may limit its dual operator to a subregion of the entanglement wedge.

%Another question is 
%about how to sharpen the concept of sub-region
%duality if we know the butterfly velocities of boundary operators.
%Given a boundary theory that has the holographic bulk dual and given a
%generic boundary operator $O$, if we map this operator back to the
%bulk in the code-subspace, obviously the bulk operator is within the
%entanglement wedge of the boundary. However, since we can decide the
%butterfly velocities of the boundary reconstructions of all the local
%bulk operators, how we can make further constraints of the bulk
%supported region of $O$ given $O$'s butterfly velocity.

The correspondence between boundary butterfly velocities and bulk causal structure may also provide a tool to determine bulk dual geometry for a boundary theory. In the context of holographic tensor networks, there is no a priori constraints on the bulk geometry. In principle, one can determine the bulk light cone if butterfly velocities of all boundary operators are known. It is a nontrivial requirement that the light cone is consistent with that of a Riemannian space-time geometry. This provides a possible way to constraint the choice of geometry for defining tensor newtork holographic mappings for a given boundary theory. For example one may define a holographic mapping with a flat space tensor network and apply it to a CFT ground state on the boundary. Our results suggest that the bulk low energy theory will not be Lorentz invariant, which suggests that the flat space is not the right choice for the tensor network representation of a CFT ground state.

Relation between the boundary theory and bulk causal structure has also been investigated in different approaches\cite{engelhardt2016towards, engelhardt2016recovering,kim2016entanglement}. It is interesting to investigate the relation of these approaches with our results. 

%Finally, in the context of the holographic tensor network, we are
%curious about how to generalize this framework to the case where the
%bulk tensor network geometry is non-static. Also, in all of the
%previous work on the holographic tensor network, there is no
%constraint on the bulk geometry. If we take into account of the
%butterfly velocities on the boundary, which is decided by the boundary
%code-subspace and the boundary Hamiltonian, whether or not we can make
%further constrains on the bulk geometry. For example, if we ask the
%bulk geometry to be Riemannien, whether the constrains are related to
%the equation of motion of the bulk geometry?

\noindent{\bf Acknowledgement.} We acknowledge helpful discussions with
Patrick Hayden, Isaac H. Kim, Aitor Lewkowycz, John Preskill, Stephen H Shenker and Michael Walter. This work is supported by the National Science Foundation through
the grant No. DMR-1151786 (XLQ and ZY) and the David and
Lucile Packard foundation (XLQ).

\appendix 
\section{Precise definition of butterfly velocity}\label{SecDef}
The definition of butterfly velocity in Sec.\ref{sec:def} involves the
limit $\Delta t\rightarrow 0$. We will make it precise using the
$\epsilon-\delta$ definition. 

Precisely, the butterfly velocity $v(O_A;\mathcal{H}_c)$ can be
defined given a boundary code subspace $\mathcal{H}_c$ and a generic
boundary operator $O_A$ acting in region $A$. The definition is that
for all $\epsilon >0$, there exists $\delta >0$, such that as long as
$\Delta t < \delta$, the following two statements are satisfied.
\begin{enumerate}
\item If $ D > \left(v(O_A;\mathcal{H}_c)+\epsilon\right) \Delta t$, then
  $\forall ~ B\in \{R|d(R,A)=D\}$, and $\forall ~ O_B $ supported in
  region $B$,
  $\forall ~|\psi_i\rangle, |\psi_j\rangle \in \mathcal{H}_c$,
  \begin{equation}
       \bra{\psi_i} [O_A,O_B(\Delta t)]\ket{\psi_j} =0
  \end{equation}
\item If $ D <\left( v(O_A;\mathcal{H}_c)-\epsilon \right) \Delta t$, then
  $\exists ~ B\in \{R|d(R,A)=D\}$, and $\exists ~ O_B $ supported in
  region $B$,
  $\exists ~|\psi_i\rangle, |\psi_j\rangle \in \mathcal{H}_c$,
  \begin{equation}
       \bra{\psi_i} [O_A,O_B(\Delta t)]\ket{\psi_j} \neq 0
  \end{equation}
\end{enumerate}

\section{Monotonicity of butterfly velocity of operators in the
  same region}\label{app:mono1}
In this section, we will prove that among all the bulk operators
located at the extremal surfaces of a spherical boundary region, the
deeper the bulk operator, the smaller the butterfly velocity of its
reconstructed operator. Our proof applies to geometry of any
dimensions with rotation symmetry and a rotation invariant boundary region, as long as it satisfies the Einstein equation(EE) and the
null energy condition(NEC).  In
Fig.\ref{mono1} we illustrate the setup for 1+1 dimensional
boundary theory with 2+1 dimensional bulk dual.

The tool we use in our proof is the null expansion in the general
relativity, which has been used to prove that the entanglement wedge
of a boundary region contains its the causal wedge if the geometry
satisfies EE and NEC\cite{wall2012maximin, hubeny2013global}. We will
introduce the important notions in this section. More details about
the null expansion and Raychaudhuri equation can be found in 
textbooks about general relativity.

Let us start with boundary region $A$, bounded by the extremal surface
$\gamma_A$. Then we shoot light-like geodesics 
perpendicular to $\gamma_A$ pointing towards the boundary
(Fig.\ref{fig:apmono}(a)). We define $U^\mu\equiv dx^\mu/d\tau$ to be
the null vector along the null surface pointing towards the
boundary. Since $U^{\mu}$ is null, we have the freedom to do affine
transformation on $\tau$. We scale $\tau$ so that when $\tau $ is a
constant, it specifies a co-dimensional 2 surface that is
perpendicular to $U^\mu$ and when $\tau=0$, the co-dimensional 2
surface is exactly $\gamma_A$. Thus we have defined a one parameter
family of co-dimensional 2 surfaces $\Gamma(\tau)$, and
$h_{\mu\nu}(\tau)$ is the induced metric on $\Gamma(\tau)$.  The null expansion $\theta$ is defined as
\begin{equation}
   \theta[\Gamma(\tau)] = h^{\mu\nu}(\tau) \nabla_\mu U_\nu  
\end{equation}
The null expansion satisfies the  Raychaudhuri equation
\begin{equation}
  \frac{d\theta[\Gamma(\tau)]}{d\tau} = -\frac{1}{d-1} \theta^2 -\sigma_{\mu\nu}\sigma^{\mu\nu} - R_{\mu\nu} U^\mu U^\nu
\end{equation}
where $\sigma_{\mu\nu}$ is the shear part of the extrinsic curvature. The important thing is that when the geometry satisfies EE, 
\begin{equation}
  R_{\mu\nu} U^\mu U^\nu = \left(8\pi G \left(T_{\mu\nu} - \frac{1}{d-1} T g_{\mu\nu}\right) + \frac{2\Lambda g_{\mu\nu}}{d-1}\right) U^\mu U^\nu = 8\pi G T_{\mu\nu}U^\mu U^\nu
\end{equation}
and NEC means $T_{\mu\nu}U^\mu U^\nu\geq 0$ if $U^\mu$ is null. Since
$\sigma_{\mu\nu}\sigma^{\mu\nu}$ is also non-negative, we conclude that
$\theta[\Gamma(\tau)]$, the null expansion, decreases monotonically
along $\tau$.
 
Now we will prove two lemmas first from which the monotonicity result can be
deduced straightforwardly. As we have mentioned $\Gamma(0)=\gamma_A$,
the extremal surface of the boundary region $A$. In our convention, $\tau>0$ ($\tau<0$) means the co-dimensional 2 surfaces are moving
towards (away from) the boundary (Fig.\ref{fig:apmono}(a)). 
\begin{figure}[ht!]
  \centering
  \includegraphics[width=0.9\textwidth]{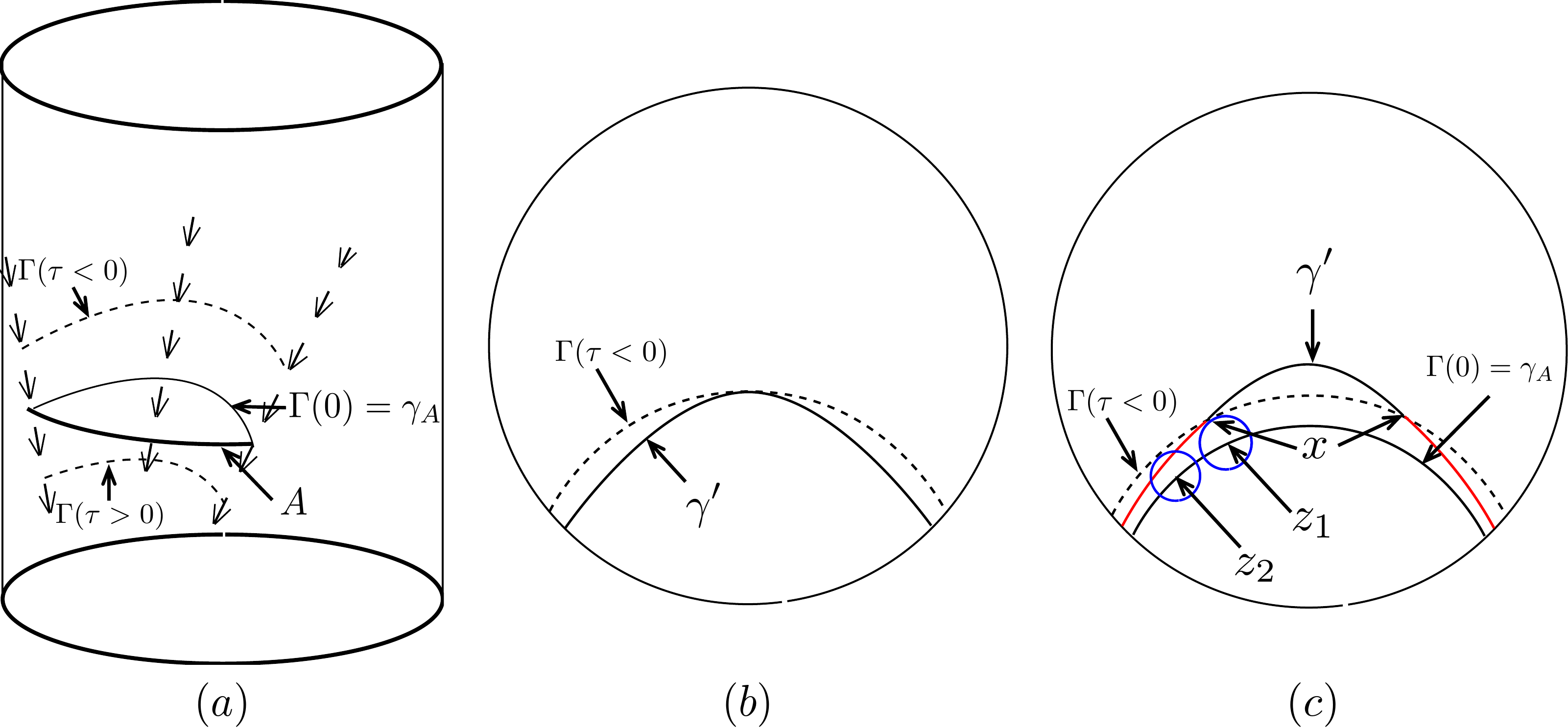}
  \caption{(a) Illustration of a null congruence of
    $\gamma_A$, the extremal surface bounding a boundary region $A$. The
    three straight lines with arrows are null geodesics emitted
    from and perpendicular to $\gamma_A$. $\Gamma(\tau>0)$ and $\Gamma(\tau<0)$
    are two co-dimensinal 2 surfaces belonging to the one parameter
    family $\Gamma(\tau)$ (see text). (b) Illustration of the setup in Lemma
    \ref{lemma1}. $\chi'$ is a geodesic surface that is tangential to $\Gamma(\tau<0)$ in figure (a). (For simplicity, only spatial directions are drawn, but the two surfaces are not required to be in a certain constant time surface.) (c) Illustration of the setup in Lemma \ref{lemma2}. $x$ is the point where $\gamma^\prime$ intersects
    with $\Gamma(\tau<0)$. $z_1$ and $z_2$ are two bulk points on
    $\gamma_A$. The blue circles are the light cones of $z_1$ and
    $z_2$. $\gamma'$ is tangential to the light cone of $z_1$ and intersects with that of $z_2$. }
  \label{fig:apmono}
\end{figure}

\begin{lemma}
  If an extremal surface $\gamma^\prime$ is tangent to the
surface $\Gamma(\tau<0)$ at its tip(Fig.\ref{fig:apmono}(b)), it will not
  intersect with the null congruence $\Gamma(\tau)$ again.\label{lemma1}
\end{lemma}
\begin{proof}
  {One useful result that we will refer to has been proven in
    \cite{wall2012maximin, hubeny2013global}.  If two co-dimensional 2
    surfaces $N_1$ and $N_2$ are tangent at $x$, and if the null
    expansions of the null congruence emitted from $N_1$ and $N_2$
    satisfy $\theta[N_1]\geq\theta[N_2]$, then in any sufficiently
    small neighborhood of $x$, $N_2$ is contained by the space-time
    region separated by the null congruence of $N_1$ towards the
    direction where the null congruence is pointing.}  Because
  $\Gamma(0)$ is the extremal surface of the boundary region $A$, the
  null expansion is $\theta[\Gamma(0)]=0$.  According to the
  Raychaudhuri equation, the null expansion decreases monotonically
  with respect to $\tau$. Thus the null expansion of $\Gamma(\tau<0)$ is 
  $\theta[\Gamma(\tau<0)]\geq 0$.  Because $\gamma^\prime$ is the
  extremal surface, $\theta[\gamma^\prime]=0$. Since
  $\gamma^\prime$ is tangent to $\Gamma(\tau<0)$ at its tip, and
  $\theta[\Gamma(\tau<0)]\geq \theta[\gamma^\prime]$, we conclude that for a
  sufficiently small region near the tip, the space-time region that
  is separated by $\Gamma(\tau)$ and contains the boundary $A$
  includes $\gamma^\prime$.

Now, we need to prove that $\gamma^\prime$ does not get out of
$\Gamma(\tau<0)$ when we are moving away from the tip. In other words,
the situation shown in Fig.\ref{fig:apmono2}(a) does not
happen. Assume that Fig.\ref{fig:apmono2}(a) does happen, then we
shrink the size of the boundary region $A$, so that $\Gamma(\tau<0)$
deforms continuously into $\tilde{\Gamma}$ and is tangent to
$\gamma^\prime$ at one point. Since $\tilde{\Gamma}$ is tangent to and
inside $\gamma'$, we have
$\theta[\tilde{\Gamma}]<\theta[\gamma^\prime]=0$. This is a
contradiction since $\tilde{\Gamma}$ is the causal future of the
extremal surface of some boundary region smaller than $A$, so that
$\theta[\tilde{\Gamma}]\geq 0$ according to the Raychaudhuri
equation. Therefore we conclude that Fig.\ref{fig:apmono2}(a) does not
happen.
\end{proof}

\begin{lemma}
For a geodesic surface
$\gamma^\prime$ which intersects with $\Gamma(\tau<0)$ at a point
$x$, then the points on $\gamma^\prime$ that are closer to the
boundary (red part in Fig.\ref{fig:apmono}(c)) than $x$ will not
intersect with $\Gamma(\tau)$.\label{lemma2}
\end{lemma}
\begin{proof}
  The proof of this lemma only requires to exclude the situation in
  Fig.\ref{fig:apmono2}(b), which can be ruled out following the same reasoning as the proof of Lemma \ref{lemma1}. 
\end{proof}
\begin{figure}[ht!]
  \centering
  \includegraphics[width=0.8\textwidth]{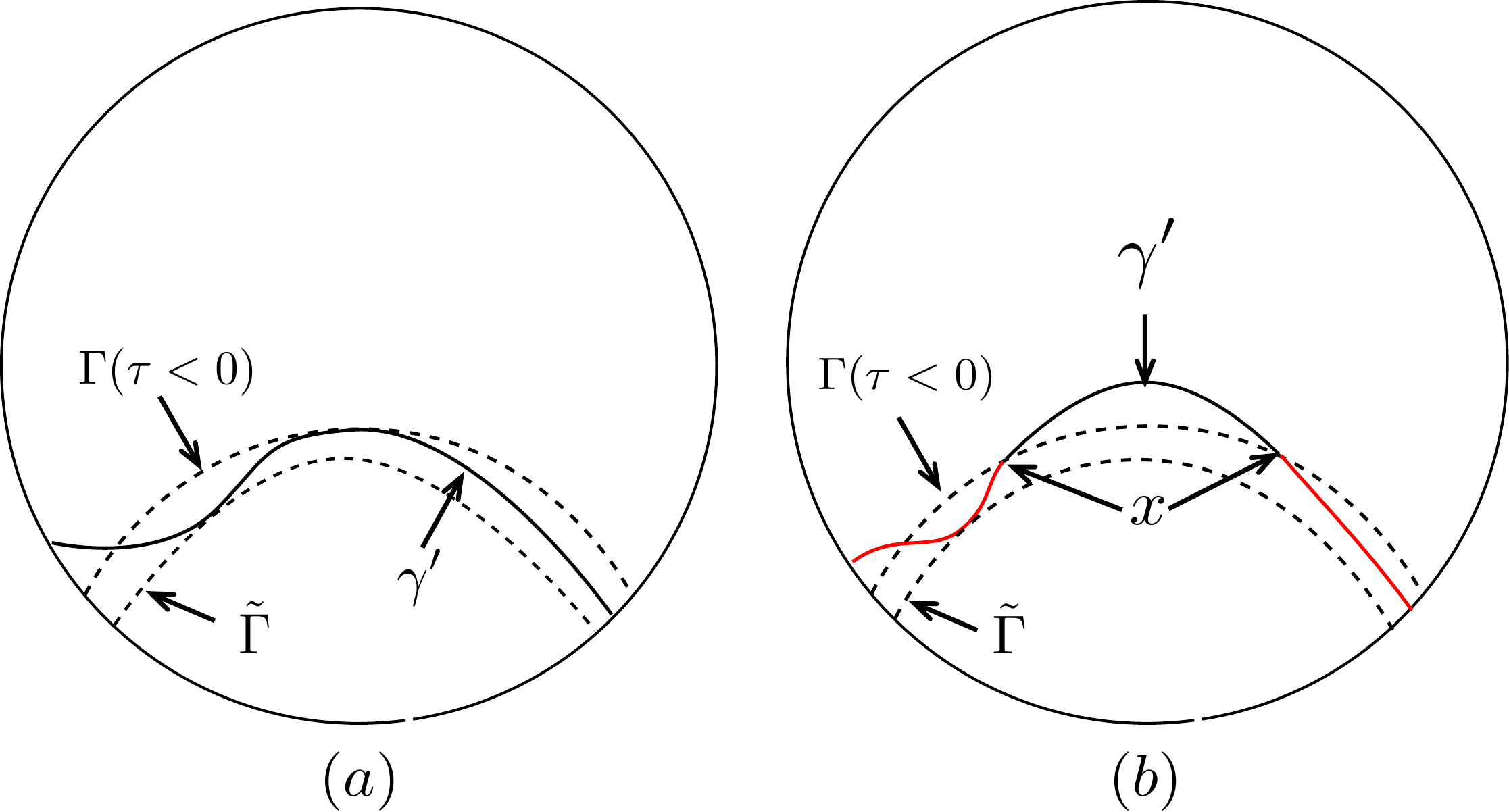}
  \caption{Two hypothetical situations that do not occur for geometries satisfying EE and NEC. $\Gamma(\tau<0)$ is the causal future of $\gamma_A$,
    the extremal surface of the boundary region $A$ (see. Fig. \ref{fig:apmono} (a)). (a) $\gamma^\prime$ is
an extremal surface which is tangential to $\Gamma(\tau<0)$. If $\gamma'$ intersects with $\Gamma(\tau<0)$ again, one can shrink the boundary region bounding $\gamma'$ and find another extremal surface $\tilde{\Gamma}$ that is tangential to $\Gamma(\tau<0)$ and is between $\Gamma(\tau<0)$ and the boundary. (b) The same argument applies to a $\gamma'$ that intersects with $\Gamma(\tau<0)$ at point $x$.  }
  \label{fig:apmono2}
\end{figure}

These two lemmas leads to the monotonicity result. { In
  Fig.\ref{fig:apmono}(c), $z_1$ and $z_2$ are two points on
  $\Gamma(0)=\gamma_A$. (It should be noted that $\gamma_A$ and $\gamma'$ are at
  different boundary time, although for the purpose of illustration we have only drawn the spatial directions.)  By construction, $\gamma^\prime$ is the minimal surface that
  is tangent to the light cone of $z_1$ (the blue circle around
  $z_1$), thus the distance between $\gamma^\prime$ and $\gamma_A$ on
  the boundary decides the butterfly velocity of $z_1$. To decide the butterfly velocity of $z_2$, we notice that $\gamma'$ intersects with the lightcone of $z_2$. Thus we must increase the size of the boundary region enclosed by $\gamma^\prime$ to find the minimal surface that is tangent to the light cone of $z_2$. Thus the butterfly velocity of $z_2$ is bigger than that of $z_1$. }

%Thus we need to expand $\chi^\prime$, so that the domain of
%dependence of $\chi^\prime$ includes more points of $\Gamma(0)$ that
%are closer to the boundary. Based on out protocol, $\chi^\prime$ lies
%on the same time slice on the boundary, thus the bigger the size of
%$\chi^\prime$, the bigger the butterfly velocity. This proves our
%monotonicity result.

\section{Proof of the saturation of butterfly velocity when the causal wedge coincides with the entanglement wedge}\label{app:proof}
\label{app:vBeqc}

In this section, the assumptions we make on the bulk geometry are that
1) it is asymptotic AdS; 2) it satisfies the Einstein equation(EE),
and null energy condition(NEC). We prove that when the causal
  wedge and the entanglement wedge of a boundary region coincide, the
  butterfly velocity of generic operators supported on this boundary region is
  $c$.

% \ZY{To make the discussion concrete, we first define the ``distance
%   function'' $\mathcal{F}$.  $\mathcal{F}(x,y)$ is a scalar function
%   between the points $x$, which is located on a causal surface
%   $\gamma$, and $y$, that is space-like separated from $x$. If $y$ is
%   outside the causal wedge whose causal surface is $\gamma$, then
%   $\mathcal{F}(x,y)$ is the proper distance between $x$ and
%   $y$. Otherwise, $\mathcal{F}(x,y)$ is negative the proper distance
%   between $x$ and $y$. Then, we will show that in the limit
%   $\Delta t\rightarrow 0$, the distance between any point on the
%   causal surface $\gamma_{u\Delta t}$ and its closest neighbor on the
%   entanglement surface $\chi_{u\Delta t}$ is of $O(\Delta t^2)$, while
%   the distance between any point on causal surfaces
%   $\gamma_{u\Delta t}$ and its closest neighbor on
%   $\gamma_{v\Delta t}$, $u\neq v$ is of $O(\Delta t)$.} \XLQ{[what is
%   $\gamma_{v\Delta t}$?]} 

\begin{figure}[ht!]
  \centering
  \includegraphics[width=0.45\textwidth]{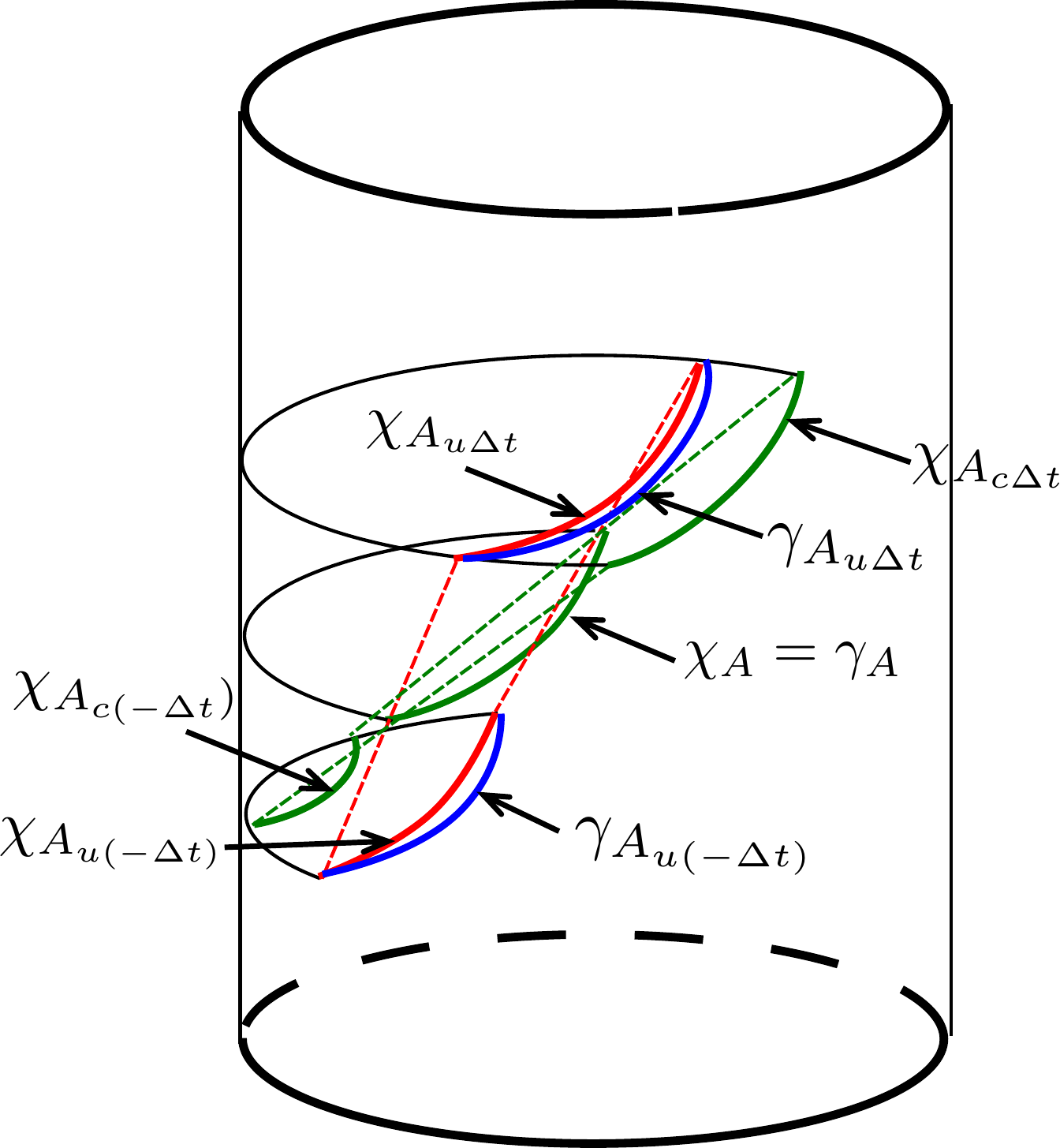}
  \caption{The setup in Sec. \ref{app:vBeqc}. The green curves are the causal surfaces of three boundary regions, $A$, its expansion $A_{c\Delta t}$ and contraction $A_{c(-\Delta t)}$, respectively.  The causal surface $\chi_A$
    of the boundary region $A$ coincides with its entanglement surface
    $\gamma_A$. The two red curves are the causal surfaces of
    $A_{u\Delta t}$, $A_{u(-\Delta t)}$ for $u<c$, respectively. The two blue
    curves are the entanglement surfaces of $A_{u\Delta t}$ and
    $A_{u(-\Delta t)}$, respectively. }
  \label{fig:causalwedgeproof}
\end{figure}

In Fig.\ref{fig:causalwedgeproof}, the green curve in the middle is
$\chi_A$, the causal surface of the boundary region $A$, which
coincides with $\gamma_A$, the entanglement surface of the boundary
region $A$. The other two green curves are the causal surfaces of
$A_{c\Delta t}$ and $A_{c(-\Delta t)}$, the expansion of boundary region $A$ by speed of light $c$ to time $\Delta t$.  The two red curves are the
causal surfaces of $A_{u\Delta t}$ and $A_{u(-\Delta t)}$ with a velocity 
$u<c$, and the two blue curves are the entanglement surfaces of
$A_{u\Delta t}$ and $A_{u(-\Delta t)}$. It has been proven that the
entanglement surfaces lie outside or coincide with the causal
surfaces\cite{wall2012maximin, hubeny2013global} for the asymptotic
AdS geometry that satisfies EE and NEC. Thus the entanglement surfaces
of the series of the boundary regions $A_{u\cdot \tau}$
$\tau\in[-\Delta t, \Delta t]/\{0\}$, will not penetrate into their
causal surfaces and at least at $\tau=0$, the two surfaces
coincide. {Now, we pick an arbitrary curve
$x(\tau),\tau\in[-\Delta t, \Delta t]$ such that $x(\tau)$ lives on
the causal surface $\chi_{A_{u\cdot \tau}}$. Correspondingly,
$y(\tau)$ is the point on the entanglement surface
$\gamma_{A_{u\cdot \tau}}$ that is closest to $x(\tau)$.  Because the
distance $d(x(\tau),y(\tau))$ is 0 at $\tau=0$ and $y(\tau)$ does not
cross $x(\tau)$ when $\tau\in[-\Delta t, \Delta t]$, the curves are tangential to each other at $\tau=0$, and we have }
\begin{equation}
  d(x(\tau),y(\tau)) = O(\tau^2),~~~ \tau\in[-\Delta t, \Delta t]
\end{equation}

{On the other hand, because $u<c$, for $\tau\in[-\Delta t,\Delta t]$,
$\chi_{u\cdot \tau}$ and $\chi_{c\cdot \tau}$ cross each other at
$\tau = 0$. Thus if we pick an arbitrary curve
$x(\tau),\tau\in[-\Delta t, \Delta t]$ such that $x(\tau)$ lives on
the causal surface $\chi_{A_{u\cdot \tau}}$, and $y(\tau)$ being the closest point to $x(\tau)$ living on  the causal surface $\chi_{A_{c\cdot \tau}}$, the distance between $x(\tau)$ and $y(\tau)$ is}
\begin{equation}
  d(x(\tau),y(\tau)) = O((c-u)\tau),~~~ \tau\in[-\Delta t, \Delta t]
\end{equation}

Finally, we put these ingredients together. $\chi_A=\gamma_A$ means if
$u<c$, the causal wedge $C_{A_{u\Delta t}}$ whose causal surface is
$\chi_{A_{u\Delta t}}$ contains no part of $\gamma_A$, because the
distance between any point on $\chi_{A_{u\Delta t}}$ and $\chi_{A_{c\Delta t}}$ is of order $O((c-u)\Delta t)$, while
the distance between any point on $\chi_{A_{u\Delta t}}$ and $\gamma_{A_{u\Delta t}}$ is of $O(\Delta t^2)$. Therefore we conclude that the butterfly velocity $v$ of a local operator at $\gamma_A$ must satisfy $v>u$ for any $u<c$. In other words, we must have $v=c$. In summary we conclude that the butterfly velocity of generic operators,
supported on the boundary region of which the causal wedge coincide
with the entanglement wedge is $c$.

\section{Flat space holography}\label{app:flat} 

We start from the $d+1$ dimensional flat space metric.
$ds^2 = -dt^2+dr^2+r^2d\theta^2+r^2\sin\theta^2d\Omega^2_{d-2}$.
Because it is symmetric in time translation, we only need to focus on
a single time slice to study the extremal surfaces. Besides, we will focus on the butterfly velocity of the boundary operators living
on spherical regions, so that the boundary region is fully characterized
by the size of the spherical cap, determined by the parameter $\Theta$. 

We set the boundary to be at $r=\Lambda$. Thus the induced boundary
metric is
$ds^2=-dt^2+\Lambda^2d\theta^2+\Lambda^2\sin\theta^2d\Omega_{d-2}^2$. At
time $t$, the extremal surface that covers the boundary region
$A=\{y| r=\Lambda, \theta\in[ 0,\Theta]\}$ is
\begin{equation}
  r_\Theta(\theta) = \frac{\Lambda\cos\Theta}{\cos\theta}, \theta\in[0,\Theta]
\end{equation}

Without loss of generality, we look at the bulk operator $\phi_x$
located on the extremal surface at point $x$, where
$\theta = \theta_x$,
$r = r_\Theta(\theta_x)= \frac{\Lambda\cos\Theta}{\cos\theta_x}$,
$\vec{\Omega}_{d-2} =\vec{0}$.  Now we find the extremal surface of
the boundary region $[0,\Theta+\Delta \Theta]$ at time $t+\Delta t$
\begin{equation}
  r_{\Theta+\Delta \Theta}(\theta)= \frac{\Lambda\cos(\Theta+\Delta \Theta)}{\cos\theta}, \theta\in[0,\Theta + \Delta\Theta]
\end{equation}

The minimum of the proper distance between the bulk point $x$ and the
extremal surface $r_{\Theta+\Delta \Theta}(\theta)$ is
\begin{equation}
  d\left(x, r_{\Theta+\Delta \Theta}(\theta) \right) = -c^2\Delta t^2 + \left(\Lambda \sin\Theta \Delta\Theta\right)^2
\end{equation}
which is independent of $\theta_x$, the angular position of $x$. 

Thus in order for the entanglement wedge of the boundary region $[0,\Theta+\Delta \Theta]$ to include the point $x$, one needs to require
\begin{equation}
  c\Delta t =\Lambda \sin\Theta \Delta\Theta
\end{equation}
so that the butterfly velocity of the boundary reconstruction of $\phi_x$ is 
\begin{equation}
  v(O_A[\phi_x];\mathcal{H}_c) = \frac{\Lambda \Delta \Theta}{\Delta t} = \frac{c}{\sin\Theta}
\end{equation}

\bibliographystyle{plain}
\bibliography{CodeSpace}

\end{document}